\documentclass[a4paper,reqno]{amsart}

\usepackage[utf8x]{inputenc}
\usepackage{hyperref}
\usepackage{xcolor}
\hypersetup{
    colorlinks,
    linkcolor={red!50!black},
    citecolor={blue!50!black},
    urlcolor={blue!80!black}
}
\usepackage{amsaddr} 

\usepackage{comment}
\usepackage{faktor}
\usepackage{cancel}

\usepackage{footmisc}

\calclayout

\usepackage{microtype}
\usepackage{MnSymbol}
\usepackage{textcomp}

\usepackage{BOONDOX-cal}

\newtheorem{Theorem}{Theorem}[subsection]
\newtheorem{Theorem*}{Theorem}
\newtheorem{Lemma}[Theorem]{Lemma}
\newtheorem{Corollary}[Theorem]{Corollary}
\newtheorem*{Corollary*}{Corollary}
\newtheorem{Proposition}[Theorem]{Proposition}
\newtheorem*{Proposition*}{Proposition}

\newtheorem{Conjecture*}[Theorem*]{Conjecture}
\theoremstyle{definition}
\newtheorem{Definition}[Theorem]{Definition}
\newtheorem*{Definition*}{Definition}
\newtheorem{Remark}[Theorem]{Remark}
\newtheorem*{Remark*}{Remark}

\newtheorem*{Example*}{Example}

\newtheorem*{Question*}{Question}

\usepackage{xparse}
\usepackage{xspace}

\newcommand{\nc}{\newcommand}
\nc{\renc}{\renewcommand}
\nc{\on}{\operatorname}
\nc{\term}[1]{#1\xspace}

\usepackage{tikz}
\usetikzlibrary{matrix}
\usepackage{tikz-cd}

\usepackage[inline]{enumitem}

\nc{\sA}{\ensuremath{\mathcal{A}}\xspace}
\nc{\sB}{\ensuremath{\mathcal{B}}\xspace}
\nc{\sC}{\ensuremath{\mathcal{C}}\xspace}
\nc{\sD}{\ensuremath{\mathcal{D}}\xspace}
\nc{\sE}{\ensuremath{\mathcal{E}}\xspace}
\nc{\sF}{\ensuremath{\mathcal{F}}\xspace}
\nc{\sG}{\ensuremath{\mathcal{G}}\xspace}
\nc{\sH}{\ensuremath{\mathcal{H}}\xspace}
\nc{\sI}{\ensuremath{\mathcal{I}}\xspace}
\nc{\sJ}{\ensuremath{\mathcal{J}}\xspace}
\nc{\sK}{\ensuremath{\mathcal{K}}\xspace}
\nc{\sL}{\ensuremath{\mathcal{L}}\xspace}
\nc{\sM}{\ensuremath{\mathcal{M}}\xspace}
\nc{\sN}{\ensuremath{\mathcal{N}}\xspace}
\nc{\sO}{\ensuremath{\mathcal{O}}\xspace}
\nc{\sP}{\ensuremath{\mathcal{P}}\xspace}
\nc{\sQ}{\ensuremath{\mathcal{Q}}\xspace}
\nc{\sR}{\ensuremath{\mathcal{R}}\xspace}
\nc{\sS}{\ensuremath{\mathcal{S}}\xspace}
\nc{\sT}{\ensuremath{\mathcal{T}}\xspace}
\nc{\sU}{\ensuremath{\mathcal{U}}\xspace}
\nc{\sV}{\ensuremath{\mathcal{V}}\xspace}
\nc{\sW}{\ensuremath{\mathcal{W}}\xspace}
\nc{\sX}{\ensuremath{\mathcal{X}}\xspace}
\nc{\sY}{\ensuremath{\mathcal{Y}}\xspace}
\nc{\sZ}{\ensuremath{\mathcal{Z}}\xspace}

\nc{\bA}{\ensuremath{\mathbf{A}}\xspace}
\nc{\bB}{\ensuremath{\mathbf{B}}\xspace}
\nc{\bC}{\ensuremath{\mathbf{C}}\xspace}
\nc{\bD}{\ensuremath{\mathbf{D}}\xspace}
\nc{\bE}{\ensuremath{\mathbf{E}}\xspace}
\nc{\bF}{\ensuremath{\mathbf{F}}\xspace}
\nc{\bG}{\ensuremath{\mathbf{G}}\xspace}
\nc{\bH}{\ensuremath{\mathbf{H}}\xspace}
\nc{\bI}{\ensuremath{\mathbf{I}}\xspace}
\nc{\bJ}{\ensuremath{\mathbf{J}}\xspace}
\nc{\bK}{\ensuremath{\mathbf{K}}\xspace}
\nc{\bL}{\ensuremath{\mathbf{L}}\xspace}
\nc{\bM}{\ensuremath{\mathbf{M}}\xspace}
\nc{\bN}{\ensuremath{\mathbf{N}}\xspace}
\nc{\bO}{\ensuremath{\mathbf{O}}\xspace}
\nc{\bP}{\ensuremath{\mathbf{P}}\xspace}
\nc{\bQ}{\ensuremath{\mathbf{Q}}\xspace}
\nc{\bR}{\ensuremath{\mathbf{R}}\xspace}
\nc{\bS}{\ensuremath{\mathbf{S}}\xspace}
\nc{\bT}{\ensuremath{\mathbf{T}}\xspace}
\nc{\bU}{\ensuremath{\mathbf{U}}\xspace}
\nc{\bV}{\ensuremath{\mathbf{V}}\xspace}
\nc{\bW}{\ensuremath{\mathbf{W}}\xspace}
\nc{\bX}{\ensuremath{\mathbf{X}}\xspace}
\nc{\bY}{\ensuremath{\mathbf{Y}}\xspace}
\nc{\bZ}{\ensuremath{\mathbf{Z}}\xspace}

\nc{\bbA}{\ensuremath{\mathbb{A}}\xspace}
\nc{\bbB}{\ensuremath{\mathbb{B}}\xspace}
\nc{\bbC}{\ensuremath{\mathbb{C}}\xspace}
\nc{\bbD}{\ensuremath{\mathbb{D}}\xspace}
\nc{\bbE}{\ensuremath{\mathbb{E}}\xspace}
\nc{\bbF}{\ensuremath{\mathbb{F}}\xspace}
\nc{\bbG}{\ensuremath{\mathbb{G}}\xspace}
\nc{\bbH}{\ensuremath{\mathbb{H}}\xspace}
\nc{\bbI}{\ensuremath{\mathbb{I}}\xspace}
\nc{\bbJ}{\ensuremath{\mathbb{J}}\xspace}
\nc{\bbK}{\ensuremath{\mathbb{K}}\xspace}
\nc{\bbL}{\ensuremath{\mathbb{L}}\xspace}
\nc{\bbM}{\ensuremath{\mathbb{M}}\xspace}
\nc{\bbN}{\ensuremath{\mathbb{N}}\xspace}
\nc{\bbO}{\ensuremath{\mathbb{O}}\xspace}
\nc{\bbP}{\ensuremath{\mathbb{P}}\xspace}
\nc{\bbQ}{\ensuremath{\mathbb{Q}}\xspace}
\nc{\bbR}{\ensuremath{\mathbb{R}}\xspace}
\nc{\bbS}{\ensuremath{\mathbb{S}}\xspace}
\nc{\bbT}{\ensuremath{\mathbb{T}}\xspace}
\nc{\bbU}{\ensuremath{\mathbb{U}}\xspace}
\nc{\bbV}{\ensuremath{\mathbb{V}}\xspace}
\nc{\bbW}{\ensuremath{\mathbb{W}}\xspace}
\nc{\bbX}{\ensuremath{\mathbb{X}}\xspace}
\nc{\bbY}{\ensuremath{\mathbb{Y}}\xspace}
\nc{\bbZ}{\ensuremath{\mathbb{Z}}\xspace}

\nc{\fp}{\mathfrak{p}}
\nc{\fq}{\mathfrak{q}}
\nc{\fj}{\mathfrak{j}}
\nc{\fm}{\mathfrak{m}}

\nc{\mrm}[1]{\ensuremath{\mathrm{#1}}\xspace}
\nc{\mit}[1]{\ensuremath{\mathit{#1}}\xspace}
\nc{\mbf}[1]{\ensuremath{\mathbf{#1}}\xspace}
\nc{\mcal}[1]{\ensuremath{\mathcal{#1}}\xspace}
\nc{\msc}[1]{\ensuremath{\mathscr{#1}}\xspace}
\nc{\mfr}[1]{\ensuremath{\mathfrak{#1}}\xspace}

\nc{\id}{\mathrm{id}}
\DeclareMathOperator{\Hom}{\on{Hom}}

\usepackage{mathtools}

\nc{\Spec}{\mrm{Spec~}}
\nc{\syn}{\mrm{syn}}
\nc{\crys}{\mrm{crys}}
\nc{\fib}{\mrm{fib}}
\nc{\Fil}{\mrm{Fil}}
\nc{\Ext}{\mrm{Ext}}
\nc{\colim}{\mrm{colim}}
\nc{\Gal}{\mrm{Gal}}
\nc{\Frob}{\mrm{Frob}}
\nc{\Tr}{\mrm{Tr}}
\nc{\ord}{\mrm{ord}}

\renewcommand{\prec}{\operatorname{prec}}
\newcommand{\softO}{O^{\sim}}

\title{Computation of classical and $v$-adic $L$-series of $t$-motives\vspace{-2mm}}
\author[X. Caruso]{Xavier Caruso}
\address{CNRS/IMB, 351 cours de la lib\'eration, Talence F-33405}
\author[Q. Gazda]{Quentin Gazda}
\address{CMLS, \'Ecole Polytechnique, Palaiseau F-91120}
\date{\today}

\makeatletter
\def\l@subsection{\@tocline{2}{0pt}{4pc}{6pc}{}}
\makeatother

\begin{document}

\maketitle

\begin{abstract}
We design an algorithm for computing the $L$-series associated
to an Anderson $t$-motive, exhibiting quasilinear complexity 
with respect to the target precision. Based on experiments, we
conjecture that, after renormalization by the local factor at
$v$, the order of vanishing at $T=1$ of the $v$-adic $L$-series 
of a given Anderson $t$-motive does not depend on the finite 
place $v$.
\end{abstract}

\setcounter{tocdepth}{1}
\tableofcontents

\setlength{\parindent}{0em}
\parskip 0.75em

\thispagestyle{empty}

Let $\bbF$ be a finite field with $q$ elements. \emph{Arithmetic in function fields} is a branch of arithmetic geometry which consists in replacing any occurrence of the ring $\bbZ$ by the polynomial ring~$A:=\bbF[t]$ and to study its ``number theory". Compared to classical arithmetic, this has the benefit of being generally easier to handle. One may play along the existence of a prime field, that of the Frobenius endomorphism and the quite convenient fact that the theory at the infinite place is non-archimedean. 

Arithmetic of function fields went so far to offer a remarkably simple analogue of the still conjectural category of mixed motives over $\bbQ$. Objects of this category--called \emph{$t$-motives over $\bbF(\theta)$}--consist of finite free modules $M$ over $\bbF(\theta)[t]$ equipped with a Frobenius-semilinear isomorphism $\tau_M$ of $M$ with poles along the diagonal ideal $(t-\theta)$ (we refer the reader to \S\ref{sec:t-motives-and-their-L-series} for details). They were introduced in a slightly different\footnote{In \emph{loc.\,cit}, Anderson considers $t$-motives that are nowadays called \emph{effective}, see Definition \ref{def:t-motive} below.} form by Anderson in his pioneer work \cite{anderson}. One may attach various realizations to a $t$-motive, including an \emph{$\ell$-adic realization} $\operatorname{T}_{\ell}\underline{M}$. Here $\ell\in \bbF[t]$ is a monic irreducible polynomial, and $\operatorname{T}_{\ell}\underline{M}$ is a finite free module over $A_\ell$, the completion of $A$ at the place $\ell$, which is equipped in addition with a continuous action of the absolute Galois group $G_{\bbF(\theta)}$ of~$\bbF(\theta)$. 

They also possess function fields valued $L$-series which are defined under no condition on the $t$-motive and which shares striking similarities with motivic $L$-functions. They are formally defined as the product 
\begin{equation}\label{eq:def-L}
L(\underline{M},T)=\prod_{\fp}{\operatorname{det}_{A_\ell}\left(\id-T^{\deg \fp}\Frob_{\fp}^{-1}|(\operatorname{T}_{\ell}\underline{M})^{I_{\fp}}\right)^{-1}} \in \bbF(t)[\![T]\!]
\end{equation}
which ranges over the monic irreducible polynomials in $\bbF[t]$; in each factor, $\ell$ is chosen distinct from~$\fp$. The group $I_{\fp}\subset G_{\bbF(\theta)}$ denotes the inertia group at $\fp$ and $\Frob_{\fp}\in G_{\bbF(\theta)}/I_{\fp}$ denotes the arithmetic Frobenius. To give a nonambiguous sense to formula~\eqref{eq:def-L}, one should check first that the principal term of the product is a polynomial in $\bbF[t][T]$ independent from $\ell$. This is done below in Corollary \ref{cor:formula-local-L}. 

Given a finite place $v$ of $\bbF[t]$, we also define a series $L_v(\underline{M},T)$ by a similar product~\eqref{eq:def-L} but removing the local factor at $\fp=v$. In what follows, we also write $L_{\infty}(\underline{M},T)$ for $L(\underline{M},T)$.

\begin{Remark*}
Similar $L$-series have already been introduced in many different contexts; for $\tau$-sheaves \cite{tagushi-wan}, $\tau$-crystals \cite{boeckle-pink,bockle}, Drinfeld modules \cite{taelman,mornev}, $t$-modules \cite{fang,angles-ngodac-tavares-ribeiro} and shtukas \cite{lafforgue}. They all compare to the above series, although we decided to remain quiet on details for the sake of concision. 
\end{Remark*}

Most of the known properties of these $L$-series are deduced from 
certain versions of the Woods Hole trace formula, as written by Anderson 
in \cite{anderson-trace} (\emph{e.g.} see Lafforgue \cite{lafforgue}, 
Taelman \cite{taelman-woodshole}, B\"ockle-Pink \cite{boeckle-pink} for 
direct applications). In the present text, we explain how Anderson's 
trace formula is also well-suited to provide a fast-converging algorithm 
to compute the $L$-series of a $t$-motive.

A crucial ingredient for this is the notion of maximal model introduced
and studied
in~\cite{gazda}. Indeed, Anderson's formula cannot be applied directly
on a $t$-motive $\underline{M}$ but requires its maximal model.
In the present article, we first design an algorithm for computing 
the maximal model. We then design a second algorithm, mostly based 
on Anderson's formula, to compute the $L$-series of $\underline{M}$
from the knowledge of its maximal model. We analyze the complexity
of this second algorithm, obtaining eventually the following theorem.

\begin{Theorem*}[\emph{cf} Theorem~\ref{thm:complexity}]
There exists an algorithm taking as input
\begin{itemize}
\item the maximal model of a
$t$-motive $\underline M$ of rank $r$ (encoded by a matrix 
defining the action of $\tau_M$), 
\item a place $v$ of $\bbF[t]$ of degree $d$,
\item a target precision $\prec$
\end{itemize} and outputs the $L$-series $L_v(\underline M, T)$ at
precision $v^{\prec}$ for a cost of
\[
\softO\left(r^\Omega \cdot \left(\frac{d_\theta}{q-1} + d\right)^\Omega \cdot d \cdot \prec \right)
\]
operations in $\bbF$ where $d_\theta$ is the maximal $\theta$-degree
of an entry of the matrix of $\tau_M$ and $\Omega$ denotes a feasible 
exponent for the 
computation of the characteristic polynomial over an abstract ring.
\end{Theorem*}

Along the way, we encountered the following theorem which seems to
be new (\emph{cf} Theorems~\ref{thm:effective-ratio} and~\ref{thm:l-adic-infinite-radius}).

\begin{Theorem*}
Let $\underline{M}$ be a $t$-motive over $\bbF(\theta)$ and let $v$ be a place of $\bbF[t]$ (finite or infinite). Let $a_n\in \bbF(t)$ denote the $n$th coefficient of $L_v(\underline{M},T)$ as a power series in $T$. Then,
\[
|a_n|_{v}=O\left(q^{-\deg v \cdot c^n}\right), \quad \text{where}\quad c=q^{1/(\operatorname{rank}\underline{M})(\deg v)}>1.
\]
Above, $|\cdot|_v=q^{-(deg v)v_v(\cdot)}$ where $v_v$ denotes the $v$-adic
valuation. 
In particular $L_v(\underline{M},T)$ has infinite $v$-adic radius of convergence as a function in $T$. If, in addition, $\underline{M}$ is effective\footnote{Meaning that $\tau_M$ is regular along the divisor $t=\theta$, see Definition \ref{def:t-motive} below.}, then $L_v(\underline{M},T)$ is a polynomial with coefficients in $\bbF[\theta]$. 
\end{Theorem*} 

The fast convergence of the coefficients is illustrated in
Figure~\ref{fig:valuations} with the tensor powers of the dual of the 
Carlitz motive over $\bbF_3(\theta)[t]$ and the place $v = t$.
(Recall that the Carlitz motive is the $t$-motive $C$
of rank $1$ and basis $e$ equipped with $\tau_C : e \mapsto
(t-\theta) e$.)
\begin{figure}
$\begin{array}{|c|c|c|c|c|c|c|c|c|c|c|}
\hline
M & v_t(a_0) & v_t(a_1) & v_t(a_2) & v_t(a_3) & v_t(a_4) & v_t(a_5) & v_t(a_6) & v_t(a_7) & v_t(a_8) & v_t(a_9) \\
\hline
C^\vee & 
  0 & 1 & 6 & 23 & 76 & 237 & 722 & 2179 & 6552 & 19673 \\
(C^\vee)^{\otimes 2} &
  0 & 0 & 4 & 20 & 72 & 232 & 716 & 2172 & 6544 & 19664 \\
(C^\vee)^{\otimes 3} & 
  0 & 3 & 18 & 69 & 228 & 711 & 2166 & 6537 & 19656 & 59019 \\
(C^\vee)^{\otimes 4} & 
  0 & 0 & 2 & 16 & 66 & 224 & 706 & 2160 & 6530 & 19648 \\
(C^\vee)^{\otimes 5} & 
  0 & 1 & 14 & 63 & 220 & 701 & 2154 & 6523 & 19640 & 59001 \\
(C^\vee)^{\otimes 6} & 
  0 & 0 & 12 & 60 & 216 & 696 & 2148 & 6516 & 19632 & 58992 \\
(C^\vee)^{\otimes 7} & 
  0 & 1 & 12 & 59 & 214 & 693 & 2144 & 6511 & 19626 & 58985 \\
(C^\vee)^{\otimes 8} & 
  0 & 0 & 10 & 56 & 210 & 688 & 2138 & 6504 & 19618 & 58976 \\
(C^\vee)^{\otimes 26} & 
  0 & 0 & 28 & 164 & 624 & 2056 & 6404 & 19500 & 58840 & 176912 \\
\hline
\end{array}$
\caption{$t$-adic valuation of the ten first coefficients of $L_t(M; T)$ ($q=3$)}
\label{fig:valuations}
\end{figure}

We hope that computations in mass using our algorithm might shed new lights on conjectural behavior of function fields special $L$-values. For instance, using it at different finite places $v$, we observed the following pattern which does not seem to be known.
\begin{Conjecture*}
\label{conjecture}
The order of vanishing of $L_v(\underline{M},T)/P_v(\underline{M},T)$ at $T=1$ is independent of the finite place $v$, where $P_v(\underline{M},T)$ denotes the polynomial local $L$-factor of $\underline{M}$ at $v$.  
\end{Conjecture*}

In most cases, the polynomial $P_v(\underline{M},T)$ has no root at $1$ and 
the order of vanishing of $L_v(\underline{M},T)$ itself is then expected to be 
constant with respect to $v$.
There exist however examples where the contribution of $P_v(\underline{M},T)$
cannot be omitted. One of them is the $t$-module $M = \bbF_2(\theta)[t] e_1 
\oplus \bbF_2(\theta)[t] e_2$ on which the action of $\tau_M$ is given by
\begin{equation}
\label{eq:example}
\begin{array}{r@{\hspace{0.5ex}}l}
\tau_M(e_1) & = (\theta+1) e_1 + (t+1) e_2 \smallskip \\
\tau_M(e_2) & = (t \theta + \theta) e_1 + (t^2 + \theta) e_2
\end{array}
\end{equation}
The orders to vanishing of $L_v(\underline{M},T)$ and $P_v(\underline{M},T)$
for all the places $v$ of degree up to $4$ are reported in the table of
Figure~\ref{fig:orders}.
\begin{figure}
$\begin{array}{|c|c|c|}
\hline
v & \ord_{T=1} P_v(\underline{M},T) & \ord_{T=1} L_v(\underline{M},T) \\
\hline
t & 1 & 3 \\
t + 1 & 0 & 2 \\
t^2 + t + 1 & 0 & 2 \\
t^3 + t + 1 & 1 & 3 \\
t^2 + t^2 + 1 & 0 & 2 \\
t^4 + t + 1 & 0 & 2 \\
t^4 + t^3 + 1 & 0 & 2 \\
t^4 + t^3 + t^2 + t + 1 & 0 & 2 \\
\hline
\end{array}$
\caption{Orders of vanishing for the $t$-motive defined by Eq.~\eqref{eq:example} ($q=2$)}
\label{fig:orders}
\end{figure}
One sees that the difference between both orders is always $2$ but
$P_v(\underline{M},T)$ vanishes at $T=1$ sometimes.

Many other examples are available online on our gitlab repository:

\begin{center}
\url{https://plmlab.math.cnrs.fr/caruso/anderson-motives/-/tree/main/examples}
\end{center}

\section{$t$-motives and their $L$-series}\label{sec:t-motives-and-their-L-series}

\subsection{$t$-motives}

Throughout this article, we fix a finite field $\bbF$ and denote by $q$ its
cardinality. We write $A = \bbF[t]$ and $K = \text{Frac } A = \bbF(t)$.
A \emph{place $\ell$ of $K$} is, by definition, a discrete valuation
subring $A_{(\ell)}$ of $K$ whose fraction field is $K$. Concretely, a
place of $K$ is either the place at infinity, denoted by $\infty$, or a
finite place given by an irreducible monic polynomial in $A$.
To $\ell$ as above, one attaches:
\begin{enumerate}
\item a discrete valuation $v_\ell$ on $K$, taking nonnegative values on
$A_{(\ell)}$,
\item its maximal ideal $\fm_\ell:=\{x\in K~|~v_{\ell}(x)>0\}$ (also denoted by $\ell$ if finite),
\item its residue field $\bbF_\ell:=A_{(\ell)}/\fm_\ell$
which is a finite extension of $\bbF$,
\item its degree $\deg \ell$ which is the degree of the field extension 
$\bbF_\ell$ over $\bbF$,
\item the completion $A_{\ell}$ of $A_{(\ell)}$, and $K_{\ell}$ of $K$.
\end{enumerate}

A ring of great importance to define $t$-motives is the tensor 
product\footnote{Throughout this article, unlabeled tensor product will 
implicitly be taken over $\bbF$.} $A\otimes A$.
By convention, we continue to write $t$ for $t \otimes 1$ and we set
$\theta = 1 \otimes t \in A \otimes A$. With this notation, $A \otimes A$
becomes isomorphic to the ring of bivariate polynomials $\bbF[t,\theta]$.
We shall also often consider the tensor product $A \otimes K$ which,
using the same convention, becomes isomorphic to $\bbF(\theta)[t]$.

We let $\tau$ be the endomorphism of $A\otimes K$ acting on $A$ by the
identity and on $K$ by the Frobenius $x \mapsto x^q$. If $M$ is a module
over $A\otimes K$, we set $\tau^\star   M = (A\otimes K) \otimes_{\tau, 
A\otimes K} M$ where the notation means that we consider $A\otimes K$ as
an algebra over itself \emph{via} $\tau$. For $m\in M$, we also note by $\tau^\star m$ the element $(1\otimes m)\in \tau^{\star}M$.

\begin{Definition}\label{def:t-motive}
An \emph{$t$-motive} is the datum $\underline{M}=(M,\tau_M)$ of a finite free module $M$ over $A \otimes K$
together with an $(A\otimes K)$-linear isomorphism
\[\textstyle
\tau_M:(\tau^\star  M)\big[\frac 1{t-\theta}\big]\stackrel{\sim}{\longrightarrow} M\big[\frac 1{t-\theta}\big].
\]
The rank of $M$ is called the \emph{rank of $\underline{M}$}. We say that $\underline{M}$ is \emph{effective} if $\tau_M$ factors as a map $\tau^\star  M\to M$. 
\end{Definition}

We fix a separably closure $K^s$ of $K$ with (profinite) Galois group $G_K$.
Let $\fp$ be a finite place of $K$.
Given a place $\wp$ of $K^s$ above $\fp$ (that is, a valuation ring 
$\sO_{(\wp)}$ of $K^s$ whose field of fractions is $K^s$ and which contains 
$\sO_{(\fp)}$ as a sub-valution ring), we introduce the following subgroups 
of $G_K$:
\begin{enumerate}
\item The decomposition subgroup $D_{\wp}:=\{\sigma\in G_K~|~\sigma(\sO_{(\wp)})=\sO_{(\wp)}\}$,
\item The inertia subgroup $I_{\wp}:=\{\sigma\in G_K~|~\sigma(\wp)=\wp\}\subset D_{\wp}$, where $\wp$ is the maximal ideal of~$\sO_{(\wp)}$.
\end{enumerate}
The group $D_{\wp}$ also identifies with the profinite Galois group 
$\Gal(K_\wp^s|K_{\fp})$ where $K_\wp^s$ denotes the completion of $K^s$ with 
respect to the valuation $v_{\wp}$. We also recall that the inclusion 
$I_{\wp}\subset D_{\wp}$ only depends on $\fp$ up to conjugation, and that it 
sits in a short exact sequence:
\[
0\longrightarrow I_{\wp}\longrightarrow D_{\wp}\longrightarrow G_{\bbF_\fp} \longrightarrow 0
\]
where $G_{\bbF_\fp}$ is the absolute Galois group of the residue field 
$\bbF_\fp$. We denote by $\Frob_{\fp} \in G_{\mathbb{F}_{\mathfrak{p}}}$
the Frobenius endomorphism, \emph{i.e.} $x\mapsto x^{q^d}$ where $d=\deg \fp$.

Let $\underline{M}$ be a $t$-motive of rank $r\geq 0$. 
Let $\ell$ be a finite place of $A$; we attach to $\underline{M}$ a continuous $\ell$-adic representation of $D_\wp$ (see  \cite[Definition 2.24]{gazda}): 
\[
\operatorname{T}_\ell \underline{M}:=\varprojlim_n (M/\ell^nM \otimes_{K} K_\wp^s)^{\tau_M=1}.
\]
Lemma 2.26 in \emph{loc.\,cit.} implies that $\operatorname{T}_\ell 
\underline{M}$ defines a rank $r$ continuous representation of $D_\wp$. As a 
consequence, the group $G_{\bbF_\fp}$ acts on $(\operatorname{T}_\ell 
\underline{M})^{I_{\wp}}$ continuously. This allows to introduce the local 
$L$-factor at $\fp$ relative to $\ell$.

\begin{Definition}
The {\it local $L$-factor of $\underline{M}$ at $\fp$ and relative to $\ell$} is the polynomial
\[
P_{\fp}(\operatorname{T}_\ell \underline{M},T):=\operatorname{det}_{A_{\ell}}\left(\id-T^{\deg \fp}\Frob_\fp^{-1}|(\operatorname{T}_\ell \underline{M})^{I_\wp}\right) \in A_\ell[T].
\]
\end{Definition}

\begin{Remark}
As the notation suggests, the polynomial $P_{\fp}(\operatorname{T}_\ell \underline{M},T)$ depends on $\fp$ (but not on~$\wp$). This is due to the fact that the determinant is invariant under conjugation. We will show next that it is also independent of $\ell$.
\end{Remark}

\subsection{A formula for the local $L$-factor}
In this subsection, we give a formula for the local $L$-factor in terms of the maximal integral model. On the one hand, this will show that the local $L$-factor of $\underline{M}$ at $\fp$ is independent of $\ell$ (as long as $\fp$ is not above $\ell$) and has coefficients in $K$ and, on the other hand, this will allow us to apply Anderson's trace formula. 

\subsubsection{Maximal models}
Let $\underline{M}=(M,\tau_M)$ be a $t$-motive. We recall the notion of lattices and models of $\underline{M}$ following \cite[\S 4]{gazda}.
\begin{Definition}
Let $N$ be a $A\otimes A$-submodule of $M$.
\begin{enumerate}
\item $N$ is called a \emph{lattice} if it is finitely generated and generates $M$ over~$K$. 
\item $N$ is called \emph{stable by $\tau_M$} (or \emph{stable} for short) if $\tau_M(\tau^\star N)\subset N[(t-\theta)^{-1}]$.
\item If $N$ is both a lattice and stable, then we say that $N$ is a \emph{model} of $\underline{M}$. 
\item If $N$ is a model of $\underline{M}$, we call $N$ \emph{maximal} if it is not strictly contained in any 
other model.
\end{enumerate}
\end{Definition}

Note that a model of $\underline{M}$ is not necessarily free; if it is, we call it a \emph{free model}. 

Models of $t$-motives are in some sense analogue to integral models of 
varieties defined over number fields. The following result, reminiscent of 
the existence of Néron models for abelian varieties, was proven in 
\emph{loc.~cit.} (combination of Proposition 4.30, Theorem 4.32 there).

\begin{Theorem}\label{thm:existence-uniqueness-model}
A maximal model for $\underline{M}$ exists and is unique. 
In addition, it is projective over~$A \otimes A$.
\end{Theorem}

We denote the maximal model of $\underline{M}$ by $M_{\sO}$.

\begin{Remark}
Since $A \otimes A \simeq \bbF[t, \theta]$ is a bivariate polynomial
ring over a field, Quillen--Suslin's theorem~\cite{quillen} shows that the maximal 
model of~$\underline{M}$ is
necessarily free. We shall give in Section~\ref{ssec:maximalmodel} a 
direct proof of this freeness property (including an algorithm for computing 
a basis) avoiding the use of the profound and complicated theorem of Quillen
and Suslin.
\end{Remark}

To a free model $N$, we associate its \emph{discriminant} $\Delta_N\in \bbF[\theta]$ defined as follows. If a basis $\mathbf{b}$ of $N$ over $\bbF[t,\theta]$ is given, the determinant of the matrix representing $\tau_M$ in $\tau^\star \mathbf{b}$ and $\mathbf{b}$ depends on the choice of $\mathbf{b}$ only up to a unit. Moreover it takes the form $a (t-\theta)^h \Delta_N$ where $a$ is a nonzero scalar in $F$ (which might depend on $\mathbf b$), $h$ is an integer which depends only on $\underline M$ and $\Delta_N$ is in $\bbF[\theta]$ and depends only on $N$. The latter is called the \emph{discriminant of $N$}.

\begin{Definition}
We call the \emph{discriminant of $\underline{M}$} and denote by $\Delta_{\underline{M}}$ the discriminant of $M_{\sO}$. We say that $\underline{M}$ has \emph{good reduction} whenever $\Delta_{\underline{M}}=1$.  
\end{Definition}

Relying on Theorem \ref{thm:existence-uniqueness-model}, one may introduce another local $L$-factor. 
Given a finite place $\fp$ of $A$, we define
\[
P_{\fp}(M_{\sO},T):=\operatorname{det}_{A[\fp^{-1}]}\left(\id-T\tau_M|M_{\sO}\otimes_{A\otimes A} (A[\fp^{-1}]\otimes \bbF_\fp)\right)\in A[\fp^{-1}][T].
\]
The determinant is well-defined as it is taken over a free $A[\fp^{-1}]$-module of rank $(\operatorname{rank}\underline{M})(\deg \fp)$. The following theorem will be proved at the end of this subsection.

\begin{Theorem}\label{thm:another-formula-local}
For any two distinct finite places $\fp$ and $\ell$ of $A$, we have
$P_{\fp}(\operatorname{T}_\ell \underline{M},T)=P_{\fp}(M_{\sO},T)$. 
\end{Theorem}

This has the following immediate corollary which implies the well-definedness of the infinite product we wrote earlier in~\eqref{eq:def-L}.
\begin{Corollary}\label{cor:formula-local-L}
The polynomial $P_{\fp}(\operatorname{T}_\ell \underline{M},T)$ is independent of $\ell$ and belongs to $1+T^{d}K[T^d]$ where~$d=\deg \fp$. 
\end{Corollary}

\subsubsection{Frobenius spaces}
Before elaborating on the proof of Theorem \ref{thm:another-formula-local}, we need to review some facts about Frobenius spaces and their models. They are much easier to handle than $t$-motives and their models, and we will prove Theorem \ref{thm:another-formula-local} by reducing to the case of Frobenius spaces. We follow \cite[\S 4.1]{gazda} closely.

We consider a finite place $\fp$ of $K$. As before, we let $\sO_\fp$ (resp.
$K_{\fp}$) be the 
completion of $A$ (resp. of $K$) at $\fp$ and we fix a separable closure $K_{\fp}^s$ of $K_{\fp}$. 
We denote by $G_\fp$ the absolute Galois group $\Gal(K_{\fp}^s|K_{\fp})$ 
equipped with its profinite topology, and we let $I_\fp\subset G_\fp$ be the 
inertia subgroup.
Let $\sigma:K_{\fp}\to K_{\fp}$ be the $q$-Frobenius.

A \emph{Frobenius space} over $K_{\fp}$ is a pair $(V,\varphi)$ where $V$ 
is a finite dimensional $K_{\fp}$-vector space and $\varphi:\sigma^\star  V\to V$ 
is a $K_{\fp}$-linear isomorphism. 
By an {\it $\sO_\fp$-lattice in $V$}, we mean a finitely generated 
$\sO_\fp$-submodule $N$ of $V$ which generates $V$ over $K_{\fp}$. We say 
that an $\sO_\fp$-submodule $N$ {\it is stable by $\varphi$} if 
$\varphi(\sigma^\star   N)\subset N$.

\begin{Definition}
Let $N$ be a finitely generated $\sO_\fp$-submodule of $V$. 
\begin{enumerate}
\item We say that $N$ is an {\it integral model for $(V,\varphi)$} if $N$ is an $\sO_\fp$-lattice in $V$ stable by $\varphi$. We say that $N$ is {\it maximal} if it is not strictly included in another integral model for $(V,\varphi)$.
\item We say that $N$ is a {\it good model for $(V,\varphi)$} if $\varphi(\sigma^\star  N)=N$. We say that $N$ is {\it maximal} if it is not strictly included in another good model of $(V,\varphi)$.
\end{enumerate}
\end{Definition}

The following was proven in Propositions 4.2 and 4.11 of \cite{gazda}.

\begin{Proposition}\label{prop:maximal}
A maximal model (resp. maximal good model) for $(V,\varphi)$ exists and is unique.
\end{Proposition}

\begin{Definition}
We denote by $V_{\sO}$ the maximal integral model of $(V,\varphi)$. We denote by $V_{\text{good}}$ the maximal good model of $(V,\varphi)$. We say that $(V,\varphi)$ has {\it good reduction} if the inclusion $V_{\text{good}}\subseteq V_{\sO}$ is an equality.
\end{Definition}

The relation among $V_{\sO}$ and $V_{\text{good}}$ can be described as follows (Propositions 4.12 and 4.14 in \emph{loc.\,cit.}). 
\begin{Proposition}\label{prop:decomposition}
The module $V_{\sO}$ decomposes as $V_{\operatorname{good}}\oplus V_{\operatorname{nil}}$, where $V_{\operatorname{nil}}\subset V$ is a sub-$\sO$-module stable by $\varphi$ and on which $\varphi$ is topologically nilpotent. 
\end{Proposition}

Frobenius spaces are directly related to Galois representations thanks
to Katz's equivalence of categories, recalled below.

\begin{Theorem}[\cite{katz-modular}, Proposition 4.1.1]
There is a rank-preserving equivalence of categories from the category
of Frobenius spaces over $K_\fp$ and the category of $\bbF$-linear 
continuous representations of $G_\fp$, explicitly given by the functor
\[
\operatorname{T} :
\underline{V}=(V,\varphi) \longmapsto 
(V\otimes_{K_{\fp}} K_{\fp}^s)^{\varphi = 1}:=
\big\{x\in V\otimes_{K_{\fp}} K_{\fp}^s~|~x=\varphi(\sigma^\star   x)\big\}
\]
where $G_{\fp}$ acts on $\operatorname{T}\underline{V}$ \emph{via} its
action on $K_{\fp}^s$.
\end{Theorem}

Let $B$ be a finite commutative $\bbF$-algebra. Noticing that a 
$B$-linear representation is nothing but an $\bbF$-linear representation endowed 
with an additional action of $B$, we deduce from Katz's theorem that the
functor $\operatorname{T}$ induces another equivalence of categories: \smallskip
\[
\left\{\text{Frobenius spaces over } K_\fp \atop
\text{with coefficients in }B\right\}
\stackrel{\sim}{\longrightarrow} 
\left\{\text{continuous~}B\text{-linear} \atop 
\text{representations~of~}G_\fp\right\}
\]
where, by a Frobenius space over $K_\fp$ with coefficients in $B$, we mean 
a finite $B\otimes K_{\fp}$-module $V$ equipped with an $B\otimes
K_{\fp}$-linear isomorphism $\varphi:\tau^\star  V\stackrel{\sim}{\to} V$ where we 
denoted $\tau=\id\otimes \sigma$ for consistency.

\subsubsection{Local $L$-factors of Frobenius spaces}

Katz's equivalence relates two objects to which one can associate a local 
$L$-factor. Our aim is to show that the $L$-factors are preserved under the 
equivalence.
Any Frobenius space $(V,\varphi)$ over $K_{\fp}$ with coefficients in $B$ 
induces a Frobenius space over $K_{\fp}$ by forgetting the action of $B$. In 
that respect, it admits a maximal integral model $V_\sO$ as well as a maximal 
good model $V_{\operatorname{good}}$. By functoriality of the assignment 
$V\mapsto V_\sO$ (resp. $V\mapsto V_{\operatorname{good}}$), both $V_\sO$ and 
$V_{\operatorname{good}}$ are canonically $B\otimes K_{\fp}$-modules.

\begin{Definition}
\begin{enumerate}
\item Let $\underline{V}=(V,\varphi)$ be a Frobenius space over $K_{\fp}$ with coefficients in $B$ with maximal integral model $V_{\sO}$. The polynomial
\[
P(\underline{V},T)=\operatorname{det}_{B}\left(\id-T\varphi |V_{\sO}/\fp V_{\sO}\right) \in 1+TB[T]
\] 
is called the {\it local $L$-factor of $\underline{V}$}. 
\item Let $\underline{\operatorname{T}}$ be a $B$-linear continuous representation of $G_\fp$. Let $d$ be the degree of $\fp$. The polynomial
\[
P(\underline{\operatorname{T}},T)=\operatorname{det}_{B}\left(\id-T^d \Frob_\fp^{-1}|\underline{\operatorname{T}}^{I_\fp}\right) \in 1+T^dB[T^d]
\] 
is called the {\it local $L$-factor of $\underline{\operatorname{T}}$}. 
\end{enumerate}  
\end{Definition}

Let $\underline{V}$ be a Frobenius space, and let $\operatorname{T}\underline{V}$ be the $B$-linear continuous representation of $G_\fp$ associated to it through Katz's equivalence. A key step to prove Theorem \ref{thm:another-formula-local} is the following proposition. 

\begin{Proposition}\label{prop:local-factor-katz}
$P(\underline{V},T)=P(\operatorname{T}\underline{V},T)$. 
\end{Proposition}

\begin{proof}
First observe that, by Katz's equivalence over $\sO_{\fp}$ (as in the proof of \cite[Proposition 4.17]{gazda}), the representation 
$\operatorname{T}\underline{V}^{I_\fp}$ corresponds to 
$(V_{\operatorname{good}}\otimes_{\sO_{\fp}}\bar{\sO}_{\fp})^{\varphi=1}$.\\
Fixing a basis of $V_{\operatorname{good}}$ as an $\sO_{\fp}$-module, and writing the action of 
$\varphi$ in this basis as a matrix $G\in 
\operatorname{Mat}_r(\sO_{\fp})$, we recognize that the reduction map 
$(V_{\operatorname{good}}\otimes_{\sO_{\fp}} \bar{\sO}_{\fp})^{\varphi=1}\to 
(V_{\operatorname{good}}\otimes_{\sO_{\fp}}\bar{\bbF}_{\fp})^{\varphi=1}$ 
reduces to a map 
$\operatorname{red}:X(\bar{\sO}_{\fp})\to X(\bar{\bbF}_{\fp})$
where
\[X(S):=\big\{(x_1,...,x_r)\in S^r|(x_1,...,x_r)-(x_1^q,...,x_r^q)G=0\big\}.\]
By the multivariate Hensel's lemma, $\operatorname{red}$ is a bijection. 
Therefore
\begin{equation}\label{eq:final-equation}
\operatorname{T}\underline{V}^{I_\fp}=(V_{\operatorname{good}}\otimes_{\sO_{\fp}}\bar{\bbF}_{\fp})^{\varphi=1},
\end{equation}
where $\Frob_\fp^{-1}$ on the left-hand side acts as $(1\otimes \sigma^{-d})=(\varphi^d\otimes 1)$ on the right-hand side.\\
On the other hand, due to Proposition \ref{prop:decomposition}, we have 
\[
V_\sO/\fp V_{\sO}\cong (V_{\operatorname{good}}/\fp V_{\operatorname{good}})\oplus (V_{\operatorname{nil}}/\fp V_{\operatorname{nil}}),
\]
and the map $\varphi$ is nilpotent once restricted to the right-hand summand. This decomposition further respects the structure of $B$-modules. Hence, 
\[
\operatorname{det}_{B}\left(\id-T\varphi |V_{\sO}/\fp V_{\sO}\right)=\operatorname{det}_{B}\left(\id-T\varphi |V_{\operatorname{good}}/\fp V_{\operatorname{good}}\right).
\]
Therefore, 
\begin{align}
P(\underline{V},T) &=\operatorname{det}_{B}\left(\id-T(\varphi\otimes_{\sO_{\fp}}\sigma) |V_{\operatorname{good}}\otimes_{\sO_{\fp}}\bbF_p\right)  \nonumber \\
\label{eq:boeckle-pink}&= \operatorname{det}_{B\otimes \bbF_\fp}\left(\id-T^d\varphi^d |V_{\operatorname{good}}\otimes_{\sO_{\fp}}\bbF_\fp\right) \\
\label{eq:ext-scalars}&= \operatorname{det}_{B\otimes \bar{\bbF}_{\fp}}\left(\id-T^d\varphi^d\otimes_{\bbF_\fp} \id |V_{\operatorname{good}}\otimes_{\sO_{\fp}}\bbF_\fp\otimes_{\bbF_\fp} \bar{\bbF}_{\fp}\right) \\
\label{eq:lang-isogeny}&= \operatorname{det}_{B\otimes\bar{\bbF}_{\fp}}\left(\id-T^d(\varphi^d\otimes_{\sO_{\fp}} \id)\otimes \id |(V_{\operatorname{good}}\otimes_{\sO_{\fp}}\bar{\bbF}_{\fp})^{\varphi=1}\otimes\bar{\bbF}_{\fp}\right) \\
\label{eq:dext-scalars}&=\operatorname{det}_{B}\left(\id-T^d(\varphi^d\otimes_{\sO_{\fp}} \id) |(V_{\operatorname{good}}\otimes_{\sO_{\fp}}\bar{\bbF}_{\fp})^{\varphi=1}\right) \\
\label{eq:lemma-hensel}&=P(\operatorname{T}\underline{V},T)
\end{align}
where in~\eqref{eq:boeckle-pink} we used \emph{e.g.} \cite[Lemma~8.1.4]{boeckle-pink}, in~\eqref{eq:ext-scalars} we extended linearly scalars through $\bbF_{\fp}\to \bar{\bbF}_{\fp}$, in~\eqref{eq:lang-isogeny} we used Lang isogeny theorem \cite[Proposition 1.1]{katz}, in~\eqref{eq:dext-scalars} we descended through $\bbF\to \bar{\bbF}_{\fp}$, and in~\eqref{eq:lemma-hensel} we used the equality~\eqref{eq:final-equation}.
\end{proof}

\subsubsection{Proof of Theorem \ref{thm:another-formula-local}}
Let $n$ be a positive integer and let $\ell$ be a finite place of $K$ distinct from~$\fp$. Let $M_{\fp}$ denote the base-change $M\otimes_K K_{\fp}$. The datum of $\underline{V}=(M_{\fp}/\ell^n M_{\fp},\tau_M)$ defines a Frobenius space over $K_\fp$ with coefficients in $A/\ell^n$. As such, it admits a maximal model $V_{\sO}$ which is a module over $A/\ell^n \otimes A$, and Proposition \ref{prop:local-factor-katz} gives the equality
\[
P(\operatorname{T}\underline{V},T)=P(\underline{V},T)
\]
as polynomials in $A/\ell^n [T]$.
Note that the left-hand side identifies with the reduction modulo $\ell^n$ of $\det_{A_{\ell}}(\id-T^{\deg \fp}\Frob_{\fp}|(\operatorname{T}_{\ell}\underline{M})^{I_{\fp}})$, while the right-hand side corresponds to $\det_{A/\ell^n}(\id-T\tau_M|V_{\sO}\otimes_{A}\bbF_{\fp})$. 

We set $M_{\sO_\fp}:=M_{\sO}\otimes_{A\otimes A}(A\otimes \sO_{\fp})$. Note that we cannot conclude yet, as $V_{\sO}$ might not correspond to $M_{\sO_{\fp}}/\ell^n M_{\sO_{\fp}}$. However, we have 
\[
M_{\sO_{\fp}}+\ell^{k_n}M_{\fp}=V_{\sO}+\ell^{k_n}M_{\fp}
\]
for some integer $k_n\leq n$ which tends to infinity as $n$ does; this 
results first from \cite[Theorem~4.55]{gazda} to identify $M_{\sO_\fp}$ with 
the maximal model of $\underline{M}_{\fp}$, then from 
\cite[Lemma~4.47]{gazda}.
Therefore, we obtain
\[
\operatorname{det}_{A_{\ell}}(\id-T^{\deg \fp}\Frob_{\fp}^{-1}|(\operatorname{T}_{\ell}\underline{M})^{I_{\fp}})\equiv \operatorname{det}_{A_{\ell}}(\id-T\tau_M|(M_{\sO}\otimes_{A\otimes A}(A\otimes \bbF_\fp))^{\wedge}_{\ell}) \pmod{\ell^{k_n}}.
\]
But since the endomorphism $\tau_M$ acting on $(M_{\sO}\otimes_{A\otimes A}(A\otimes \bbF_\fp))^{\wedge}_{\ell}$ is induced by its $A[\fp^{-1}]$-linear action on $M_{\sO}\otimes_{A\otimes A}(A[\fp^{-1}]\otimes \bbF_\fp)$, we get  
\[
\operatorname{det}_{A_{\ell}}(\id-T^{\deg \fp}\Frob_{\fp}^{-1}|(\operatorname{T}_{\ell}\underline{M})^{I_{\fp}})\equiv \operatorname{det}_{A[\fp^{-1}]}(\id-T\tau_M|M_{\sO}\otimes_{A\otimes A}(A[\fp^{-1}]\otimes \bbF_\fp)) \pmod{\ell^{k_n}}.
\]
Taking the limit as $n$ tends to infinity achieves the proof of Theorem \ref{thm:another-formula-local}.

\subsection{Anderson's trace formula}\label{subsec:anderson-trace-formula}
We first recall the notion of \emph{nuclear operator} introduced by Anderson in \cite{anderson-trace} and its extension with general coefficient rings by B\"ockle--Pink in \cite{boeckle-pink}.
Let $B$ be an $\bbF$-algebra, let $D$ be a free $B$-module and let $\kappa:D\to D$ be a $B$-linear map. 
By freeness, one can write $D=B\otimes D_0$ (as $B$-modules) for some 
$\bbF$-vector space $D_0$. Upon the choice of $D_0$, B\"ockle and Pink
introduce the following definition.

\begin{Definition}
An $\bbF$-vector space $W_0\subset D_0$ is called \emph{a nucleus for $\kappa$} if it is finite dimensional and there exists an exhaustive increasing filtration of $D_0$ by finite dimensional $\bbF$-vector spaces $W_0\subset W_1\subset W_2\subset \cdots$ such that $\kappa(B\otimes W_{i+1})\subset B\otimes W_i$ for all $i>0$. \\
If a nucleus for $\kappa$ exists, we call $\kappa$ \emph{nuclear}.
\end{Definition}

\begin{Proposition}[\cite{boeckle-pink}, Propositions 8.3.2 \& 8.3.3]
Suppose that $\kappa$ admits a nucleus $W_0$. Then, the expressions
\begin{align*}
\Tr(\kappa) :=&\Tr_B(\kappa|B\otimes W_0) \in B, \\
\Delta(\id-T\kappa):=&\operatorname{det}_B(\id-T\kappa|B\otimes W_0)\in B[T], 
\end{align*}
are independent of the nucleus $W_0$.
\end{Proposition}

We will call $\Tr(\kappa)$ and $\Delta(\id-T\kappa)$ the \emph{trace} and the \emph{dual characteristic polynomial of $\kappa$}. 
\begin{Remark}
Note that although the trace (resp. dual characteristic polynomial) of $\kappa$ does not depend on the nucleus $W_0$, this might still depend on the decomposition $D=B\otimes D_0$. 
\end{Remark}

We express Anderson's trace formula in the form recorded by B\"ockle--Pink in \cite[\S 8.3]{boeckle-pink}. Let $B$ and $R$ be $\bbF$-algebras and assume that $R$ is a Dedekind domain. We denote by $\sigma:R\to R$ the $q$-Frobenius and $\tau=\id\otimes \sigma$ the $B$-linear ring endomorphism of $B\otimes R$ which acts as the $q$-Frobenius on $R$. 

Let $\Omega_{R}^1= \Omega_{R/\bbF}^1$ be the $R$-module of K\"ahler 
differentials of $R$ relative to $\bbF$. The divided $q$-Frobenius 
``$\frac{\sigma}{q}$"$:\Omega_{R}^1\to \Omega_{R}^1$, formally defined 
by the rule $a{\cdot}db\mapsto a^{q}b^{q-1}{\cdot}db$, defines a 
$R$-linear map $\sigma^*\Omega_{R}^1\to \Omega_R^1$. It is 
well-known that the latter induces an isomorphism between
$\sigma^*\Omega_{R}^1$ and the first group of de Rham cohomology
of $R$. Inverting it, we get a map $C:\Omega^1_R\to \sigma^\star
\Omega^1_R$ which vanishes on exact differentials; it is the
so-called \emph{$q$-Cartier operator}. Equivalently, seen as a 
semilinear map, $C$ is uniquely determined by the following 
requirements for any $a\in R$:

\begin{enumerate}
\item $C(da)=0$,
\item $C(a^{q-1}da)=da$.
\end{enumerate}

Let $N$ be a finite free $B\otimes R$-module and let $\tau_N:\tau^\star  N\to N$ be a linear morphism. To the datum of $(N,\tau_N)$, one associates another $(N^{\star},\tau_N^{\star})$. The $B\otimes R$-module $N^{\star}$ is the \emph{K\"ahler dual} of $N$:
\[
N^{\star}=\Hom_{B\otimes R}(N,B\otimes \Omega_{R}^1).
\]
The morphism $\tau_N^{\star}:N^{\star}\to \tau^\star  N^{\star}$ is given on $f\in N^{\star}$ by the expression:
\[
\tau_N^{\star}(f):=C\circ f\circ \tau_N.
\]
We fix an identification $N=B\otimes N_0$ for some finite $R$-module $N_0$. This induces an identification $N^{\star}=B\otimes N_0^{\star}$ where $N_0^{\star}=\Hom_R(N_0,\Omega_{R}^1)$. It follows from \cite[Proposition~8.3.8]{boeckle-pink} that $\tau_N^\star $ is nuclear and admits a nucleus $W_0\subset N_0^{\star}$. 

\begin{Theorem}[\cite{boeckle-pink}, Theorem 8.3.10, Anderson's trace formula]\label{thm:ATF}
In the ring $B[\![T]\!]$, we have
\[
\prod_{\fp}{\operatorname{det}_B(\id-T\tau_N| N\otimes_R R/\fp)^{-1}}=\Delta(\id-T \tau_N^{\star})
\]
where the product on the left runs over all maximal ideals $\fp$ in $R$ and the dual characteristic polynomial on the right is taken with respect to the decomposition $N^{\star}=B\otimes N_0^{\star}$.
\end{Theorem}

\begin{Remark}
Since $\operatorname{det}_B(\id-T\tau_N| N\otimes_R R/\fp)=\operatorname{det}_{B\otimes \bbF_\fp}(\id-T^d\tau_N| N\otimes_R R/\fp)$, where $d$ is the degree of $\fp$, and since there are finitely many maximal ideals in $R$ with a given degree, the formal product does make sense in $B[\![T]\!]$. 
\end{Remark}

\subsection{Global $L$-series}
Let $\underline{M}$ be a $t$-motive. For a place $v$ of $K$ (finite or infinite), we define
\[
L_v(\underline{M},T):=\prod_{\fp\neq v}{P_{\fp}(M_{\sO},T)^{-1}},
\]
where the formal product is taken over finite places distinct from $v$ (the condition is then empty if $v$ is infinite). From Corollary \ref{cor:formula-local-L}, the product makes 
sense in $K[\![T]\!]$ but also in $\sO_{v}[\![T]\!]$ when $v$ is finite. 
Moreover, it follows from Theorem~\ref{thm:another-formula-local} that $L_v(\underline{M},T)$
defined as above agrees with the $L$-series defined in the introduction
(see Equation~\eqref{eq:def-L}).

The decisive advantage of the second definition is that it can be approached
using Anderson's trace formula with the pair $(M_{\sO},\tau_M)$.
This is not actually entirely correct because $\tau_M$ does not induce an 
$\bbF[t]$-linear endomorphism of $M_{\sO}$ in full generality. This however works well
when $\underline{M}$ is effective in which case we get the following theorem.

\begin{Theorem}\label{thm:effective-ratio}
Let $\underline{M}$ be an effective $t$-motive over $K$. Then 
$L_v(\underline{M},T)$ is a polynomial with coefficients in $A$.
\end{Theorem}

\begin{proof}
In the situation where $\underline{M}$ is effective, we have 
\[
P_{\fp}(M_{\sO},T)=\operatorname{det}_{A[\fp^{-1}]}\left(\id-T\tau_M|M_{\sO}\otimes_{A\otimes A}(A[\fp^{-1}]\otimes \bbF_\fp)\right)=\operatorname{det}_{A}\left(\id-T\tau_M|M_{\sO}\otimes_{A\otimes A}(A\otimes\bbF_\fp)\right)
\]
as $\tau_M$ factors as an endomorphism of $M_{\sO}\otimes_{A\otimes A}(A\otimes \bbF_\fp)$. Therefore, we can apply Anderson's trace formula to the pair $(M_{\sO},\tau_M)$, with $B=A$ and $R=A$ if $v$ is infinite or $R=A[v^{-1}]$ if $v$ is finite, which implies that $L_v(\underline{M},T)$ is a polynomial with coefficients in $A$.
\end{proof}

The effective case is unfortunately not the most interesting one.
In the next sections, we will adapt the above strategy to the general 
case by working modulo an arbitrary positive power of $v$.

\section{Algorithm for computing the $L$-series}\label{sec:algorithm}
Let $\underline{M}$ be a $t$-motive over $K=\bbF(\theta)$ and let $v$ be a place of $K$ (finite or infinite). There are two main steps to compute the $L$-series of $\underline{M}$. First, one should explicitly construct its maximal model $M_{\sO}$ and second one should find an appropriate nucleus of $\tau_M^\star$ modulo $v^n$. In this section, we give algorithms for both steps.

\subsection{Determination of the maximal model}
\label{ssec:maximalmodel}

According to Theorem \ref{thm:existence-uniqueness-model}, $\underline{M}$ 
admits a maximal model $M_{\sO}$ which, in addition, is finite projective 
over $\bbF[t,\theta]$. By the Quillen--Suslin theorem, we know that $M_{\sO}$ 
is free over $\bbF[t,\theta]$.
In this subsection, we provide an algorithm computing a basis of
$M_{\sO}$.

\subsubsection{Construction of the maximal model}

We begin by giving an explicit iterative construction of~$M_{\sO}$.
It proceeds in several steps.

\paragraph{Step 1.} 

The first step consists in finding a model of $\underline{M}$. This is quite 
classical: it suffices to take any basis of $M$ over $\bbF(\theta)[t]$ and let $N'$
be the $\bbF[t,\theta]$-module it generates. It may not be a model, but since 
it is a lattice, there exists $f\in \bbF[\theta]\setminus\{0\}$ for which 
$\tau_M(\tau^\star N')\subset \frac{1}{f}N'[(t-\theta)^{-1}]$. Then, the~$\bbF[t,\theta]$-module $N:=fN'$ is a model of $\underline{M}$ (\emph{e.g.} 
\cite[Proposition 4.20]{gazda}). Moreover, it is free over~$\bbF[t,\theta]$.

Let $\Delta(\theta)\in \bbF[\theta]$ be the discriminant of $N$. 
We factor it into monic irreducible polynomials as 
\[
\Delta(\theta)=\fp_1(\theta)^{n_1}\cdots \fp_s(\theta)^{n_s}.
\]
Then
\[
N\big[\fp_1^{-1}, \ldots, \fp_s^{-1}\big] = 
M_{\sO}\big[\fp_1^{-1}, \ldots, \fp_s^{-1}\big].
\]
The strategy of the rest of the proof is to approach $M_{\sO}$ from $N$ 
by removing the $\fp_i$'s one by one.

\paragraph{Step 2.} Let $\fp:=\fp_s$ and $N_0:=N$. 
By induction, for $i\geq 0$, we define $\bbF[t,\theta]$-modules
\[
N_{i+1}:=\left\{x\in \fp^{-1}N_i ~\bigg|~\tau_M(\tau^\star x)\in N_i\left[\frac{1}{t-\theta}\right]\right\} \subset M.
\]
We gather some properties of this sequence in the next lemma. 
\begin{Lemma}\label{lem:properties-N_i}
The following holds.
\begin{enumerate}
\item\label{item:inclusion} 
The sequence $(N_i)_{i\geq 0}$ is increasing for the inclusion. 
\item\label{item:is-a-model} 
For all $i\geq 0$, $N_i$ is a model of $\underline{M}$.
\item\label{item:union} 
The sequence $(N_i)_{i\geq 0}$ is stationary.
\item\label{item:contained-in-max} 
Its limit, $N_{\infty}$, is a model which satisfies 
\[
N_{\infty}\subset M_{\sO} \subset 
\fp^{-1}N_{\infty}\big[\fp_1^{-1},\ldots,\fp_{s-1}^{-1}\big].
\]
\end{enumerate}
\end{Lemma}
\begin{proof}
We prove~\eqref{item:inclusion} by induction on $i$. It is clear that $N_0\subset N_1$. Assuming $N_{i-1}\subset N_i$ for some $i$, and taking $x\in N_i$, we get $\tau_M(\tau^\star x)\in N_{i-1}[(t-\theta)^{-1}]\subset N_i[(t-\theta)^{-1}]$ and hence $x\in N_{i+1}$. That is $N_i\subset N_{i+1}$.
We deduce both from~\eqref{item:inclusion} that $N_i$ is stable and that it is generating (as it contains $N_0$). This implies~\eqref{item:is-a-model}. \\
By maximality we then get $N_i\subset M_{\sO}$ and hence, by noetherianity, the sequence $(N_i)_{i\geq 0}$ stabilizes, proving~\eqref{item:union}.\\
It remains to show~\eqref{item:contained-in-max}. Because it equals $N_I$ for $I$ large, $N_{\infty}$ is a model. Let $n$ be the minimal integer such that 
$M_{\sO}\subset
\fp^{-n}N_{\infty}\big[\fp_1^{-1},\ldots,\fp_{s-1}^{-1}\big]$
and suppose that $n\geq 2$. There exists $m\in M_{\sO}$ 
such that $\fp^{n}m\in N_{\infty}\big[\fp_1^{-1},\ldots,\fp_{s-1}^{-1}\big]$
but $\fp^{n-1}m$ does not. Its image under $\tau_M$ is so that
\[
\tau_M(\tau^\star \fp^{n-1}m)=\fp^{q(n-1)}\tau_M(\tau^\star m)\in 
N_{\infty}\left[\frac{1}{t-\theta},\fp_1^{-1},\ldots,\fp_{s-1}^{-1}\right].
\]
In particular, there exists $e\in \bbF[\theta]$ supported at 
$\{\fp_1,\ldots,\fp_{s-1}\}$, such that $\tau_M(\tau^\star \fp^{n-1}em)\in 
N_{\infty}$. By design, this implies $\fp^{n-1}em\in N_{\infty}$ and thus
$\fp^{n-1}m\in N_{\infty}\big[\fp_1^{-1},\ldots,\fp_{s-1}^{-1}\big]$.
This is in contradiction with the choice of $m$.
\end{proof}

\paragraph{Step 3.} Let $L_0=\fp^{-1}N_{\infty}$. Starting from $L_0$, we define another sequence of $\bbF[t,\theta]$-modules $(L_i)_{i\geq 0}$ by the expression
\[
L_{i+1}:=\left\{x\in L_{i}~\bigg|~\tau_M(\tau^\star x)\in L_{i}\left[\frac{1}{t-\theta}\right]\right\} \subset M.
\]
We record properties of $(L_i)_{i\geq 0}$ in a lemma.
\begin{Lemma}\label{lem:properties-L_i}
For all $i\geq 0$, 
\begin{enumerate}
\item\label{item:decreasing} 
$N_{\infty}\subset L_{i+1}\subset L_{i}$,
\item \label{item:not-so-far} 
$\fp L_{i}\subset L_{i+1}$,
\item\label{item:pLi-model} 
$\fp L_i$ is a model; in particular $\fp L_i \subset M_{\sO}$,
\item\label{item:contains-maxmod} 
$M_{\sO} \subset L_i\big[\fp_1^{-1},\ldots,\fp_{s-1}^{-1}\big]$,
\item\label{item:station} the sequence $(L_i)_{i\geq 0}$ is stationary,
\item\label{item:intersect} if $L_{\infty}$ denotes the limit, then $L_{\infty}$ is a model of $\underline{M}$ and
\[
L_{\infty}\big[\fp_1^{-1},\ldots,\fp_{s-1}^{-1}\big]
=M_{\sO}\big[\fp_1^{-1},\ldots,\fp_{s-1}^{-1}\big].
\]
\end{enumerate}
\end{Lemma}
\begin{proof}
Point~\eqref{item:decreasing} follows from an immediate induction, using that $N_{\infty}$ is a model of $\underline{M}$. Point~\eqref{item:not-so-far} for $i=0$ is clear. Assuming $\fp L_{i-1}\subset L_i$ for $i>0$ and taking $x\in \fp L_i$, we can write $x=\fp\cdot l_i$ for $l_i\in L_i$ so that
\[
\tau_M(\tau^\star x)=\fp^q\tau_M(\tau^\star  l_i)\in \fp^{q}L_{i-1}\left[\frac{1}{t-\theta}\right] \subset L_{i}\left[\frac{1}{t-\theta}\right],
\]
and hence $x\in L_{i+1}$ as desired. 
Given that $\fp^q L_{i-1} \subset \fp L_i$, the previous computation shows
also that $\fp L_i$ is a model, hence point~\eqref{item:pLi-model}.
\\
Point~\eqref{item:contains-maxmod} follows by induction as well: initialization is Lemma~\ref{lem:properties-N_i}.\eqref{item:contained-in-max}. Assuming that the result holds at $i>0$, and given $x\in M_{\sO}$, 
\[
\tau_M(\tau^\star x)\in M_{\sO}\left[\frac{1}{t-\theta}\right]
\subset L_{i}\left[\frac{1}{t-\theta},\fp_1^{-1},\ldots,\fp_{s-1}^{-1}\right],
\]
and hence $x\in L_{i+1}\left[\fp_1^{-1},\ldots,\fp_{s-1}^{-1}\right]$ 
(more precisely, there exists $e\in \bbF[\theta]\setminus\{0\}$ supported at $\{\fp_1,\ldots,\fp_{s-1}\}$ such that $\tau_M(\tau^\star e x)\in L_{i}[(t-\theta)^{-1}]$, hence $ex\in L_{i+1}$). \\
To show~\eqref{item:station}, it suffices to notice that the sequence $(L_i)_{i\geq 0}$ is a decreasing sequence of lattices which is bounded below by $N_{\infty}$. It remains to prove~\eqref{item:intersect}. There exists a large enough integer $I\geq 0$ such that $L_i=L_{i+1}$ for all $i\geq I$. Let $L_{\infty}:=L_{I}$ be the common value. The module $L_{\infty}$ is a lattice in $M$ and, given $x\in L_{\infty}$, we have $x\in L_{I+1}$, and thus 
\[
\tau_M(\tau^\star x)\in L_I\left[\frac{1}{t-\theta}\right]=L_{\infty}\left[\frac{1}{t-\theta}\right].
\]
Therefore $L_{\infty}$ is stable and hence a model of $\underline{M}$. 
Hence $L_\infty \subset M_{\sO}$, implying \emph{a fortiori} that
$L_{\infty}\big[\fp_1^{-1},\ldots,\fp_{s-1}^{-1}\big]
\subset M_{\sO}\big[\fp_1^{-1},\ldots,\fp_{s-1}^{-1}\big].$
The reverse inclusion follows from~Lemma \ref{lem:properties-L_i}\eqref{item:contains-maxmod}. 
\end{proof}

We may now iterate the construction.
By Lemma~\ref{lem:properties-L_i}, $L_\infty$ is a model of
$\underline{M}$ which agrees with $M_{\sO}$ after inverting $\fp_1,
\ldots, \fp_{s-1}$. We can then set $N = L_{\infty}$ and go back to 
step $2$.

\subsubsection{Algorithm and proof of freeness}
\label{sssec:algo:maximalmodel}

We now explain how the construction presented above can be turned into an
actual algorithm. As a byproduct, we shall obtain an alternative proof
of the freeness of $M_{\sO}$ avoiding the use of the Quillen-Suslin theorem.

Without loss of generality we can assume that $\underline M = (M, \tau_M)$ is effective as any model of $(M, \tau_M)$ is a model
of $(M, (t-\theta)^h \tau_M)$ for any $h$, and \emph{vice versa}. The
maximal model of $(M, (t-\theta)^h \tau_M)$ then does not depend on~$h$.

We proceed step by step.

\paragraph{Step 1.}

We assume that $\underline{M} = (M, \tau_M)$ is given by the matrix
$\Phi \in \textrm{Mat}_r\big(\bbF(\theta)[t]\big)$ of $\tau_M$ in some basis 
$\mathbf{b}$.
If $f$ is a common denominator of the entries of $\Phi$, the 
$\bbF[t,\theta]$-module spanned by $f \mathbf{b}$ is a model $N$ of~$M$. 
From now on, we replace $\mathbf b$ by $f \mathbf b$, which amounts to
replacing $\Phi$ by $f^{q-1} \Phi$.
Besides, we have
\[
\det \Phi = (t - \theta)^h \cdot \Delta(\theta)
\]
where $h$ is some integer and $\Delta(\theta)=\Delta_N(\theta)$ is the discriminant
of $N$. One can easily determine the latter by computing the
determinant of $\Phi$ and dividing by $t - \theta$ as much as possible.

We then factor $\Delta(\theta)$ and find the places $\fp_1, \ldots,
\fp_s$ we shall work with afterwards.

\paragraph{Step 2.}

We need to explain how to compute a basis of $N_{i+1}$, knowing a 
basis $\mathbf{b}_i$ of $N_i$. 
Since $N_i$ is a model, it is stable under $\tau_M$, from what we
deduce that $\tau_M$ induces a \emph{semilinear} map $\fp^{-1} N_i \to 
\fp^{-q} N_i$ and so, a second semilinear map
$\fp^{-1} N_i/N_i \to \fp^{-q} N_i/N_i$.
Besides, by definition of $N_{i+1}$, we have
\[
N_{i+1}/N_i = 
\ker\left( \tau_M : \fp^{-1} N_i/N_i \longrightarrow \fp^{-q} N_i/N_i\right).
\]
The domain $\fp^{-1} N_i/N_i$ is obviously a free module of rank $r$ over 
$\bbF_\fp[t]$ with basis $\mathbf{b}_i^{(1)} = \fp^{-1} \mathbf{b}_i \text{ mod } N_i$. 
Similarly 
$\fp^{-q} N_i/N_i$ is free module of rank $r$ over 
$\bbF[t] \otimes \bbF[\theta]/\fp(\theta)^q$. There is moreover a computable
ring isomorphism $\bbF_\fp[u]/u^q \cong \bbF[\theta]/\fp(\theta)^q$: it 
takes $u$ to $\fp(\theta)$ and any element $a \in \bbF_\fp$ to $b^q$ where 
$b$ is a lifting in $\bbF[\theta]$ of $a^{1/q} \in \bbF_\fp$.
Therefore $\bbF[t] \otimes \bbF[\theta]/\fp(\theta)^q$ can be computationally 
identified with $\bbF_\fp[t,u]/u^q$.

Noticing the canonical embedding $\bbF_\fp[t] \hookrightarrow 
\bbF_\fp[t,u]/u^q$, we conclude that $\fp^{-q} N_i/N_i$ naturally 
appears as a free module of rank $qr$ over $\bbF_\fp[t]$ with
basis
\[
\mathbf{b}^{(q)}_i = 
\bigsqcup_{v=1}^q \big(\fp^{-v} \mathbf{b}_i \text{ mod } \fp^{-v+1} N_i\big).
\]
For these structures, the map 
$\tau_M : \fp^{-1} N_i/N_i \to \fp^{-q} N_i/N_i$ is semilinear with
respect to the endomorphism $\tau_\fp : \bbF_\fp[t] \to \bbF_\fp[t]$ 
taking $t$ to itself and $a \in \bbF_\fp$ to $a^q$.

Let $T$ be the $qr \times r$ matrix representing $\tau_M$ in the bases
$\mathbf{b}_i^{(1)}$ and $\mathbf{b}_i^{(q)}$. We write the Smith decomposition of
$T$, namely $T = USV$. Here $U \in \operatorname{SL}_{qr}(\bbF_\fp[t])$,
$V \in \operatorname{SL}_{r}(\bbF_\fp[t])$ and 
$S = (s_{jk})_{1 \leq j \leq qr, 1 \leq k \leq r}$ is a rectangular 
$qr \times r$ matrix with $s_{jk} = 0$ whenever $j \neq k$.
Moreover, all classical algorithms for determining Smith forms
present the transformation matrices $U$ and $V$ as a product of
transvection matrices. Hence, we can not only compute them but we
can also compute liftings of them, $\hat U$ and $\hat V$, lying in
$\operatorname{SL}_{qr}(\bbF[t,\theta])$ and $\operatorname{SL}_{r}(\bbF[t,\theta])$ 
respectively.

For $j \in \{1, \ldots, r\}$, we set $\delta_j = 1$ if $s_{jj} = 0$
and $\delta_j = \fp$ otherwise. Then, we claim that the columns 
of the product matrix
\[ 
\hat W := \hat V \cdot 
\left(\begin{matrix} \delta_1 \\ & \ddots \\ & & \delta_r \end{matrix}\right)
\]
is the change-of-basis matrix from $\mathbf{b}_i$ to a basis $\mathbf{b}_{i+1}$ of 
$N_{i+1}$. In other words, the columns of the $\hat W$ gather the coordinates 
in $\mathbf{b}_i$ of the basis $\mathbf{b}_{i+1}$ we were looking for.
Indeed, the columns of $\hat V$ corresponding the indices $j$ with $s_{jj}
= 0$ form a basis of $N_{i+1}/N_i$, while the remaining columns form a
basis of a complement of it in $\fp^{-1}N_i/N_i$.

\paragraph{Step 3.}

It is quite similar to Step 2. 
By Lemma \ref{lem:properties-L_i}.\eqref{item:pLi-model}, 
we know that $\fp L_i$ is stable by $\tau_M$.
Consequently, $\tau_M$ takes $L_i$ to $\fp^{-q+1} L_i$ and it gives rise
to a semilinear map $L_i/\fp L_i \longrightarrow \fp^{-q+1} L_i/\fp L_i$.
The quotient space $L_{i+1}/\fp L_i$ then appears as a kernel as follows:
\[
L_{i+1}/\fp L_i = 
\ker\left( \tau_M : L_i/\fp L_i \longrightarrow \fp^{-q+1} L_i/\fp L_i\right).
\]
We can now reuse the machinery of Step 2 to compute an explicit basis of
$L_{i+1}$ starting from a basis of $L_i$.

\begin{Remark}
Building upon this algorithm, one may compute the discriminant of 
$\underline{M}$ as well and, in particular, determine the set of places
of bad reduction of $\underline{M}$.
\end{Remark}

\subsubsection{About complexity}

It is not easy to give precise estimations on the complexity of the
algorithm depicted above because it looks difficult to control the 
size of the elements in $\bbF[t,\theta]$ that show up throughout the
computation.
One can nevertheless give bounds on the number of iterations after
the sequences $(N_i)_i$ and $(L_i)_i$ stabilize. We start with $(L_i)_i$,
which is the simplest.

\begin{Lemma}
We have $L_r = L_\infty$ where $r=\operatorname{rank}\underline{M}$.
\end{Lemma}

\begin{proof}
We deduce from Lemma~\ref{lem:properties-L_i} that $\fp L_0 = N_\infty
\subset L_i \subset L_0$ for all $i$. The lemma follows after noticing
that the $\bbF_\fp[t]$-length of the quotient $L_0 / \fp L_0$ is~$r$.
\end{proof}

In order to prove a similar bound for $(N_i)$, we need estimations
on discriminants. We recall that, 
if $N$ is a model of $\underline M$, we write $\Delta_N$ for its
discriminant.

\begin{Lemma}\label{lem:discriminant}
Let $N$ and $N'$ be two models of $\underline M$ with
$\fp N \subset N' \subset N$. Then
\[
\Delta_N(\theta) = \Delta_{N'}(\theta)
\cdot \fp(\theta)^{-(q-1)\operatorname{rank}_{\bbF_\fp[t]}(N/N')}.
\]
\end{Lemma}

\begin{proof}
We apply the snake lemma to the diagram
\[
\begin{tikzcd}
0\arrow[r] & \tau^\star N'\arrow[hook,d] \arrow[r,"\tau_M"] & (t-\theta)^{-h}N' \arrow[hook,d] \arrow[r] & \operatorname{coker}(\tau_M|N') \arrow[d] \arrow[r] & 0 \\
0\arrow[r] & \tau^\star N \arrow[r,"\tau_M"] & (t-\theta)^{-h}N \arrow[r] & \operatorname{coker}(\tau_M|N) \arrow[r] & 0 
\end{tikzcd}
\]
where $h$ is a large enough integer for which $\tau_M$ is well-defined. 
If $\alpha$ denotes the map between the two cokernels, we have the
following exact sequences
\[
0 \longrightarrow \ker\alpha \longrightarrow \operatorname{coker}(\tau_M|N') \longrightarrow \operatorname{coker}(\tau_M|N) \longrightarrow \operatorname{coker} \alpha \longrightarrow 0
\]
\[
0 \longrightarrow \ker\alpha \longrightarrow 
\tau^\star (N/N') \longrightarrow N/N' \longrightarrow \operatorname{coker} \alpha \longrightarrow 0
\]
which give the desired formula upon taking Fitting ideals.
\end{proof}

\begin{Corollary}
Let $N$ be the model of $\underline M$ constructed at Step 1 of the
algorithm and write
\[
\Delta_N(\theta)=\fp_1(\theta)^{n_1}\cdots \fp_s(\theta)^{n_s}
\]
where $\fp_1, \ldots, \fp_s$ are distinct places.
Then, when treating the prime $\fp_j$, we have
$N_\infty = N_{\lfloor \frac{n_j}{q-1} \rfloor}$.
\end{Corollary}

\begin{proof}
The algorithm starts with the place $\fp_s$. 
By Lemma~\ref{lem:discriminant}, the $\fp_s$-adic valuation of 
$\Delta_{N_i}$ decreases at least by $q-1$ at each iteration. The number 
of iterations is then bounded above by $\lfloor \frac{n_s}{q-1} \rfloor$.\\
Then, when passing to the next place, the initial model $N$ will be 
replaced by the model $L_\infty$ we have ended with.
In order to repeat the argument, we need to prove that the 
$\Delta_{L_\infty}$ and $\Delta_{N}$ only differ by a power of $\fp_s$.
This follows again from Lemma~\ref{lem:discriminant}, given that 
$\fp L_\infty \subset N_\infty \subset L_\infty$.
\end{proof}

\subsubsection{The local case}
\label{sssec:algo:localmaximalmodel}

Having a global maximal model defined over $\bbF[t,\theta]$ is very
nice. However, for many applications, it is often enough to know a
\emph{local} maximal model, involving only one place $\fp$.
Formally, given a finite place $\fp$, we introduce the ring
\[
\bbF[t,\theta]^{\wedge}_{\fp(\theta)} := 
\varprojlim_n \big(\bbF[t] \otimes \bbF[\theta]/\fp(\theta)^n\big)
\]
and extend the morphism $\tau$ to it. We note that
$\bbF[t,\theta]^{\wedge}_{\fp(\theta)}$ is isomorphic to
$\bbF_\fp[t][\![u]\!]$ where $u$ corresponds to $\fp(\theta)$ as 
in \S \ref{sssec:algo:maximalmodel}. In particular, an element
of $\bbF[t,\theta]^{\wedge}_{\fp(\theta)}$ is invertible if and 
only if its image in $\bbF_\fp[t]$ is.

A $t$-motive $\underline M = (M, \tau_M)$ defines by scalar extension 
to $\bbF[t,\theta]^{\wedge}_{\fp(\theta)}$ a pair $(M_\fp, \tau_{M_\fp})$.
There is a notion of maximal model in this new setting, and one checks
that the maximal model of $(M_\fp, \tau_{M_\fp})$ is simply
$M_{\sO,\fp} := \bbF[t,\theta]^{\wedge}_{\fp(\theta)} \otimes_{\bbF[t,\theta]}
M_{\sO}$ (\emph{e.g.} \cite[Theorem 4.55]{gazda}).

Of course, the knowledge of $M_{\sO,\fp}$ is much weaker than that
of $M_{\sO}$. However, it is enough to decide if $\underline M$ has
good reduction at $\fp$ and, more generally, to determine the local
factor at $\fp$.
Indeed, since $P_{\fp}(M_{\sO,\fp},T)=P_{\fp}(M_{\sO},T)$, we have the following theorem which is an immediate consequence
of Theorem~\ref{thm:another-formula-local}.

\begin{Theorem}
For any two distinct finite places $\fp$ and $\ell$ of $A$, we have
$P_{\fp}(\operatorname{T}_\ell \underline{M},T)=P_{\fp}(M_{\sO,\fp},T)$. 
\end{Theorem}

On the other hand, we claim that the algorithm presented in \S \ref{sssec:algo:maximalmodel} 
works similarly with $(M_\fp, \tau_{M_\fp})$ with two important 
simplifications:
\begin{itemize}
\item there is a unique place which matters, namely $\fp$,
\item the matrix $\hat V$ can be any lift of $V$ in 
$\operatorname{Mat}_r(\bbF[t,\theta]^\wedge_\fp)$;
any such lift will be automatically invertible (since its determinant
is) and will do the job.
\end{itemize}
As a consequence of the last point, we no longer need to compute 
$\hat V$ at the same time of $V$ and, more importantly, we can choose
its entries as polynomials in $t$ and $\theta$ of controlled degrees.
This results in a faster algorithm with which one can hope for nice bounds
on the complexity.

\subsection{Computation of the $L$-series}
Let $v$ be a place of $K$ (finite or infinite). In this subsection, we explain how to compute the $L$-series of $\underline{M}$ using Anderson's trace formula.  

As we observed in Subsection \ref{subsec:anderson-trace-formula}, Anderson's trace formula does not apply readily to $\underline{M}=(M,\tau_M)$. This is because $\tau_M$ does not generally define an $\bbF[t]$-linear endomorphism of $M$. Instead, the strategy consists in applying the trace formula modulo a power of $v$. 

For a better clarity, we write $\tau_{\theta}=\tau$ and we equip $\bbF[t,\theta]$ with another endomorphism $\tau_t$ given as the $\bbF[\theta]$-linear morphism $\tau_t(t)=t^q$. We observe that $\tau_t$ and $\tau_\theta$ commute and that their composition is the $q$-Frobenius on $\bbF[t,\theta]$.

We let $\bbF(t,\theta)$ be the fraction field of $\bbF[t,\theta]$; the endomorphisms $\tau_t$ and $\tau_\theta$ uniquely extend to~$\bbF(t,\theta)$. 

\subsubsection{The map $\tau_M^{\star}$}
Under the identification $\Omega_{\bbF[\theta]}^1=\bbF[\theta]d\theta$, we express the action of the $\bbF[t]$-linear extension of the $q$-Cartier operator $C=C_{\theta}:\bbF[t,\theta]\to \bbF[t,\theta]$ as follows:
\[
C_{\theta}\left(\sum_{i,j\geq 0}{a_{i,j}t^i\theta^j}\right)=\sum_{i,j\geq 0}{a_{i,qj+q-1}t^i\theta^j}.
\]
In other words, we solely select in the polynomial expansion the coefficients whose degree in $\theta$ has residue $q-1$ modulo $q$. Note that $C_\theta$ has the following action on the degree in $\theta$:
\begin{equation}\label{eq:degree-cartier}
\text{if} \quad \deg_\theta p(t,\theta)\leq d,\quad \text{then}\quad \deg_{\theta}C_{\theta}(p(t,\theta))\leq \frac{d}{q}-\frac{q-1}{q}.
\end{equation}
Implicitly, we made the identification $\sigma^\star  \Omega_{\bbF[\theta]}^1=\Omega_{\bbF[\theta]}^1$ at the cost of making $C_{\theta}$ only $\tau_{\theta}^{-1}$-semilinear; \emph{i.e.} it verifies the relation $C_{\theta}(\tau_{\theta}(x)y)=xC_{\theta}(y)$ for all $x,y\in \bbF[t,\theta]$. This observation allows to extend the Cartier operator to the fraction field of $\bbF[t,\theta]$ by mean of the formula
\[
C_{\theta}\left(\frac{x}{y}\right)=\frac{C_{\theta}(xy^{q-1})}{\tau_t(y)}. 
\]

Let $M^{\star}$ denote the algebraic dual of $M$, \emph{i.e.} $M^{\star}=\Hom_{\bbF[t,\theta]}(M,\bbF[t,\theta])$. As in Subsection \ref{subsec:anderson-trace-formula}, the map $\tau_M$ defines a \emph{dual map} $\tau_M^{\star}$ via the formula
\[
\tau_M^{\star}:M^{\star}\longrightarrow M^{\star}\otimes_{\bbF[t,\theta]} \bbF(t,\theta) , \quad f\longmapsto C_{\theta}\circ f\circ \tau_M
\]
(it is only $\bbF[t]$-linear). One should be aware that $\tau_M^{\star}$ will not preserve $M^{\star}$ nor $M^{\star}[(t-\theta)^{-1}]$ in general, explaining why Anderson's trace formula does not apply directly to $(M,\tau_M)$. 

We may circle this issue depending on whether the place $v$ is finite or infinite. Let $n$ be a positive integer.
\begin{enumerate}
\item If the place $v$ is finite, $\tau_M$ defines an $A/v^n$-linear endomorphism of
\begin{equation}\label{eq:module-to-apply-trace-finite-place}
M_{\sO}\otimes_{A\otimes A} \left(A/v^n\otimes A[v^{-1}]\right)
\end{equation}
(as $(t-\theta)$ becomes invertible in $A/v^n\otimes A[v^{-1}]$, \emph{e.g.} \cite[Lem. 3.15]{gazda}). Setting $B=A/v^n$ and $R=A[v^{-1}]$, we may apply Anderson's trace formula to the free $B\otimes R$-module~\eqref{eq:module-to-apply-trace-finite-place}.
\item If, otherwise, $v=\infty$ is the infinite place, then we proceed as follows. We introduce the rings
\[
\sA_{\infty}(A):=\varprojlim_{i} (\sO_{\infty}/\fm_{\infty}^i \otimes A)=\bbF[\theta][\![1/t]\!], \quad \sB_{\infty}(A):=K_{\infty}\otimes_{\sO_{\infty}}\sA_{\infty}(A)=\bbF[\theta](\!(1/t)\!).
\]
Next, we fix a finite free $\sA_{\infty}(A)$-submodule $\Lambda_{\infty}$ of $\sB_{\infty}(M_\sO):=M_{\sO}\otimes_{A\otimes A}\sB_{\infty}(A)$ (\emph{e.g.} the $\sA_{\infty}(A)$-module generated by a basis of $M_{\sO}$ over $\bbF[t,\theta]$). There exists $k$ large enough for which $t^{-k}\tau_M(\tau^\star \Lambda_{\infty})\subset \Lambda_{\infty}$. Setting $B=\sO_{\infty}/\fm_{\infty}^n$ and $R=A$, we may now apply Anderson's trace formula to the pair $(\Lambda_{\infty}/\fm^n_{\infty}\Lambda_{\infty},t^{-k}\tau_M)$. 
\end{enumerate}

\subsubsection{Nucleus for $\tau_M^{\star}$}
\label{sssec:nucleus}

We should now find a nucleus respectively to the above situations. Let $h\in \bbZ$ be an integer for which
\[
\tau_M(\tau^\star M)\subset \frac{1}{(t-\theta)^h}M
\]
(typically the smallest one). We let $h_0:=h$ and define $(h_i)_{i\geq 0}$ by induction using the formula $h_{i+1}=\lceil \frac{h_i}{q}\rceil$.

We start with the case of a finite place $v$.
Let $a\in \bbF_v \otimes 1$ be the image of $t$ under the quotient map $A\to \bbF_v$. We interpret $a$ in $A_{v}\otimes A$, so that $v(\theta)=(\theta-a)(\theta-a^q)\cdots (\theta-a^{q^{d-1}})$ where $d=\deg v$.\\
Let $c \geq 0$ be an integer, which we assume to be divisible by $d$ for
simplicity.

\begin{Lemma}
There exist $d$ integers $k_0, \ldots, k_{d-1}$ such that
\[ w_0 := q k_1 - k_0 - h_c,\,\, w_1 := q k_2 - k_1,\,\, w_2 := q k_3 - k_2,\,\, \ldots,\,\, w_{d-2} := q k_{d-1} - k_{d-2},\,\, w_{d-1} := q k_0 - k_{d-1}\]
are all in the range $[0, q{-}1]$.
\end{Lemma}

\begin{proof}
We set $k_0 = \left\lceil \frac{h_c}{q^d - 1} \right\rceil$ and
$k_i = \left\lceil \frac{q^{d-i} h_c}{q^d - 1} \right\rceil$ for $1 \leq i \leq d-1$.
A direct computation shows that
\[
w_i = q \left\lceil \frac{q^{d-i-1} h_c}{q^d - 1} \right\rceil - 
\left\lceil \frac{q^{d-i} h_c}{q^d - 1} \right\rceil
\]
for all $i$. The lemma follows after noticing that 
$q \lceil x \rceil - \lceil qx \rceil$ is in $[0, q{-}1]$ for all real 
number~$x$.
\end{proof}

We form the products
\begin{align*}
\delta & := \frac 1 {v(\theta)} \cdot 
\prod_{i=0}^{d-1} \left(a^{q^i} - \theta\right)^{-k_i} \cdot 
\prod_{i=1}^{c}{\left(t^{q^i} - \theta\right)^{-h_i}}, \\
\rho & := v(\theta)^{q-1} \cdot 
\prod_{i=0}^{d-1} \left(a^{q^i} - \theta\right)^{w_i} \cdot 
\prod_{i=0}^{c-1}{\left(t^{q^i}- \theta\right)^{qh_{i+1}-h_i}}.
\end{align*}
By definition, $\delta$ and $\rho$ verify the following relation:
\begin{equation}\label{eq:fundamental-relation}
\frac 1{(t-\theta)^h} 
\cdot \left( \frac{t^{q^c} - \theta}{a - \theta} \right)^{h_c} \delta 
= \rho \cdot \tau_{\theta}(\delta).
\end{equation}

Fix a basis $(e_1,\ldots,e_r)$ of $M_{\sO}$ over $\bbF[t,\theta]$ and denote by $(e_1^{\star},\ldots,e_r^{\star})$ the dual basis of~$M^{\star}$. We denote by $(b_{ij})_{1\leq i,j\leq r}\in \operatorname{Mat}_r(\bbF[t,\theta])$ the matrix representing $(t-\theta)^h\tau_M$ in $(e_1,\ldots,e_r)$. For indices $i,j\in \{1,\ldots,r\}$ and $x\in \bbF(t,\theta)$, we have 
\begin{equation}\label{eq:tauMstar}
\tau_M^{\star}(xe_i^{\star})(e_j)
= C_{\theta}\left(\frac{xb_{ij}}{(t-\theta)^h}\right). 
\end{equation}

Let $R=\bbF[\theta][v(\theta)^{-1}]$ be as above.
\begin{Lemma}\label{lem:formula-for-cartier}
For all $x\in \bbF(t,\theta)$, we have
\[
\tau_M^{\star}\left(x \delta e_i^{\star}\right)  \equiv
\sum_{j=1}^r C_{\theta}\left(\rho \cdot x \cdot b_{ij}\right) \delta e_j^{\star}
\pmod{v(t)^{q^c} \otimes R}.
\]
\end{Lemma}

\begin{proof}
Since $d$ divides $c$, the fraction 
$\frac{t^{q^c} - \theta}{a - \theta}$ is in $1 + v(t)^{q^c} \otimes R$.
The lemma follows easily by combining Equations~\eqref{eq:fundamental-relation} 
and~\eqref{eq:tauMstar}.
\end{proof}

To apply Anderson's formula, we put ourselves in the situation of Theorem \ref{thm:ATF} with ${B=\bbF[t]/ v(t)^{q^c}}$. We consider the finite free $B\otimes R$-module
\[
N:=M_{\sO}\otimes_{A\otimes A}(B\otimes R),
\]
a basis of which being $(e_1,\ldots,e_r)$.

\begin{Lemma}
The element $\delta$ is a unit in $B \otimes R$.
\end{Lemma}

\begin{proof}
We observe that
\[
\delta \equiv \frac 1 {v(\theta)} \cdot 
\prod_{i=0}^{d-1} \left(a^{q^i} - \theta\right)^{-k_i} \cdot 
\prod_{i=1}^{c}{\left(a^{q^i} - \theta\right)^{-h_i}}
\pmod {v(t) \otimes R}.
\]
Hence $\delta$ is a divisor of a power of $v(\theta)$ in 
$\bbF[t]/v(t) \otimes R \simeq \bbF_v[\theta][1/v(\theta)]$. 
It is then a unit in this ring. By Hensel's lemma, it is also a
unit in $B \otimes R$.
\end{proof}

It follows from the lemma that $(\delta e_1, \ldots, \delta e_r)$ is
also a basis of $N$. We let $N_0\subset N$ be the free $R$-module 
generated by $(\delta e_1,\ldots, \delta e_r)$ and let 
$N_0^{\star}\subset N^{\star}$ be the $R$-module generated by $(\delta e_1^{\star},\ldots, \delta e_r^{\star})$, so that we have the identifications
\[
N=B\otimes N_0, \quad \text{and} \quad N^{\star}=\Hom_{B\otimes R}(N,B\otimes \Omega_{R}^1)=B\otimes N_0^{\star}. 
\]

\begin{Lemma}
Let $s_{\max}:=\lceil \frac{d_{\theta}}{q-1}\rceil+2d+c$ where $d_{\theta}=\max_{i,j}{\deg_{\theta}(b_{ij})}$. Then, the $\bbF$-subvector space of $N_0^{\star}$ given by 
\begin{equation}
\label{eq:nucleusfinite}
W_0:=\left\langle \left\{\theta^s\cdot \delta e_j^{\star}~\bigg |~s\in \{0,\ldots,s_{\max}\}, ~j\in \{1,\ldots,r\}\right\} \right\rangle_{\bbF}
\end{equation}
is a nucleus for $\tau_M^{\star}:N^{\star}\to N^{\star}$.  
\end{Lemma}

\begin{proof}
Given two nonnegative integers $a$ and $b$, we set
\[
N_{a,b}^\star  := 
\left\langle \left\{v(\theta)^{-\alpha}\theta^{\beta}\cdot \delta e_j^{\star}~\bigg |~\alpha\in \{0,\ldots,a\}, ~\beta \in \{0, \ldots, b\}, ~j\in \{1,\ldots,r\}\right\} \right\rangle_{B}.
\]
The union of the $N_{a,b}^\star $ is clearly the entire space $N^\star $.
Applying Lemma~\ref{lem:formula-for-cartier}
with $x = v(\theta)^{-\alpha} \theta^\beta$, we get the congruence
\begin{align*}
\tau_M^{\star}\left(v(\theta)^{-\alpha} \theta^\beta \delta e_i^{\star}\right) 
& \equiv
\sum_{j=1}^r C_{\theta}\left(\rho \cdot v(\theta)^{-\alpha} \theta^\beta \cdot b_{ij}\right) \delta e_j^{\star} \\
& \equiv
v(\theta)^{-\alpha'} \sum_{j=1}^r C_{\theta}\left(\rho \cdot v(\theta)^{q\alpha'-\alpha} \theta^\beta \cdot b_{ij}\right) \delta e_j^{\star} 
\pmod{v(t)^{q^c} \otimes R}
\end{align*}
We choose $\alpha' = \left\lfloor \frac \alpha q \right\rfloor$. The
exponent $q\alpha' - \alpha$ is then in the range $[1-q, 0]$, showing that 
the expression inside $C_\theta$ is a polynomial in $\bbF[t,\theta]$
whose $\theta$-degree does not exceed
$(q-1)(2d + c) + \beta + d_{\theta}$.
Applying $C_{\theta}$ then 
gives by \eqref{eq:degree-cartier} a polynomial whose $\theta$-degree 
is at most $\frac{\beta + d_\theta}{q}+\frac{q-1}{q}(2d+c)$.
Hence, we get the inclusion 
\[
\tau_M^\star \big(N_{a,b}^\star \big) \subset N_{a',b'}^\star 
\quad \text{with} \quad
a' = \left\lfloor \frac a q \right\rfloor, \quad
b' = \left\lfloor \frac{b + d_\theta}{q}+\frac{q-1}{q}(2d+c) \right\rfloor
\]
This shows that, starting from $x \in N_{a,b}^\star $, applying
$\tau_M^\star $ enough times will eventually yields a result in
$N^\star _{0, s_{\max}} = B \otimes W_0$.
\end{proof}

For the place $v=\infty$, the method is very similar. We take for $\Lambda_{\infty}$ the free $\bbF[\theta][\![1/t]\!]$-module generated by the basis $(e_1,\ldots,e_r)$ so that $t^{-(d_t-h)}\tau_M$ stabilizes $\Lambda_{\infty}$, where $d_t=\max_{i,j}{\deg_t (b_{ij})}$. This time, we rather consider the infinite products
\[
\delta:=\prod_{i=1}^{\infty}{\left(1-\frac{\theta}{t^{q^i}}\right)^{-h_i}}, \quad \rho:=\prod_{i=0}^{\infty}{\left(1-\frac{\theta}{t^{q^i}}\right)^{qh_{i+1}-h_i}},
\]
which now satisfy the relation
\[
\frac{\delta}{\left(1-\frac{\theta}{t}\right)^h}=\rho \tau_{\theta}(\delta).
\]
The relevant nucleus for $t^{-(d_t-h)}\tau_M^\star $ acting on $\Hom_{\bbF[\theta][\![1/t]\!]}(\Lambda_{\infty},\bbF[\theta][\![1/t]\!]/(1/t)^{q^c})$ is now the finite dimensional $\bbF$-vector space 
\begin{equation}
\label{eq:nucleusinfinite}
W_0:=\left\langle \left\{\theta^s\cdot \delta e_j^{\star}~\bigg |~s\in \{0,\ldots,s_{\max}\}, ~j\in \{1,\ldots,r\}\right\} \right\rangle_{\bbF}.
\end{equation}
with $s_{\max} = c+\left\lfloor\frac{d_{\theta}}{q-1}\right\rfloor$

\subsubsection{The algorithm}

The above discussion translates readily to an algorithm computing the
$L$-series of a $t$-motive $\underline M$ at some arbitrary
precision, provided that we know a maximal model.
The algorithm is outlined in Figure~\ref{algo:Lseries}.

\begin{figure}
\noindent\hspace{1cm}%
\begin{minipage}{\textwidth - 2cm}
\emph{Input}:
\begin{itemize}
\item
the matrix $\Phi \in \operatorname{Mat}_r\big(\bbF[t,\theta]\big)$ giving the action 
of $(t - \theta)^h \tau_M$ on a maximal model
$M_{\sO}$ of $\underline M$
\item
a place $v$ of $\bbF[t]$, represented either by the symbol $\infty$ or by
an irreducible polynomial
\item
a target precision $\prec$
\end{itemize}

\smallskip

\emph{Output}:
\begin{itemize}
\item
the $L$-series $L_v(\underline M, T)$ at precision $v^{\prec}$
\end{itemize}

\bigskip

\emph{Case of a finite place}

\medskip

\textbf{1.}
Compute the smallest integer $c$ such that $d$ divides $c$ and $q^c \geq \prec$

\medskip

\textbf{2.}
Form the nucleus $W_0$ defined by Equation~\eqref{eq:nucleusfinite}

\medskip

\textbf{3.}
Compute the matrix $\Phi^\star  \in \operatorname{Mat}_{\dim_{\bbF} W_0}(\sO_v)$ giving the 
action of $\tau_M^\star $ on $\sO_v \otimes W_0$ using the explicit formula 
of Lemma~\ref{lem:formula-for-cartier}

\medskip

\textbf{4.}
Return $\det(1 - T \Phi^\star ) + O(v^{\prec})$

\bigskip

\emph{Case $v = \infty$}

\medskip

\textbf{1.}
Compute the smallest integer $c$ such that
$q^c - r \cdot (d_t - h) \cdot \big(c + \frac{d_\theta}{q-1}\big)
\geq \prec$

\smallskip

where $d_t$ (resp. $d_\theta$) is the maximal $t$-degree (resp.
$\theta$-degree) of an entry of $\Phi$

\medskip

\textbf{2.}
Form the nucleus $W_0$ defined by Equation~\eqref{eq:nucleusinfinite}

\medskip

\textbf{3.}
Compute the matrix $\Phi^\star  \in \operatorname{Mat}_{\dim_{\bbF} W_0}(\sO_v)$ 
giving the action of $t^{-(d_t-h)} \tau_M^\star $ on $\sO_v \otimes W_0$ using the explicit 
formula of Lemma~\ref{lem:formula-for-cartier}

\medskip

\textbf{4.}
Compute $\chi(U) = \det(1 - U \Phi^\star ) + O\big(t^{-q^c}\big)$
and return $\chi(t^{d_t-h} T) + O(t^{-\prec})$

\end{minipage}

\caption{Algorithm for computing $L$-series of a $t$-motive}
\label{algo:Lseries}
\end{figure}

Its correctness follows from Anderson's trace formula together with what
we have done in \S \ref{sssec:nucleus}.

Before discussing the complexity, we need to clarify how we represent
the ring $\sO_v$ and how we perform computations in it.
By definition $\sO_v$ is the completion of $\bbF[\theta]$ at the 
place $v$. Since $\sO_v$ has equal characteristic, it can be identified
with a ring of Laurent series over the residue field $\bbF_v$.
Concretely an isomorphism realizing this identification is, for
example:
\[
\iota_v : \sO_v \stackrel\sim\longrightarrow \bbF_v[\![u]\!],
\quad \theta \mapsto a + u
\]
where $a \in \bbF_v$ is a root of $v(\theta)$.
Computing $\iota_v(f(\theta))$ for $f(\theta) \in \bbF[\theta]$ then
amounts to shifting the polynomial~$f$, which can be achieved in quasilinear
time in the degree of $f$. Similarly, all field operations (addition,
subtraction, multiplication, division) in $\bbF_v[\![u]\!]$ at precision
$u^{\prec}$ can be performed in quasilinear time in the precision.

We recall that computing a characteristic polynomial of a $d \times d$ 
matrix over $\bbF[\theta]$ with entries known at precision $\prec$ can be 
achieved for a cost of $\softO(d^\Omega \cdot \prec)$ operations in $\bbF$ 
with $\Omega <2.69497$ using~\cite{matrices:kaltofen-villard}.
In what precedes, we used the standard notation $\softO$ meaning that 
logarithmic factors are hidden.

\begin{Theorem}\label{thm:complexity}
Algorithm of Figure~\ref{algo:Lseries} requires at most
\[
\softO\left(r^\Omega \cdot \left(\frac{d_\theta}{q-1} + d\right)^\Omega \cdot d \cdot \prec \right)
\]
operations in $\bbF$.
\end{Theorem}

\begin{proof}

Computing the matrix $\Phi^\star $ in step 3 amounts to computing the products 
$\rho{\cdot}b_{ij}$ in $\bbF_v[\![u]\!]$ (for $i,j$ varying between $1$ and $r$) 
and selecting the relevant coefficients.
This can be achieved for a cost of $\softO(r^2 d \cdot \prec)$ operations
in $\bbF$.\\
Coming back to the definition, we see that $\dim_{\bbF} W_0 \in
O\big(r\cdot \big(\frac{d_\theta}{q-1} + d + c\big)\big)$. Hence the
computation of the characteristic polynomial in step 4 requires no more
than 
\[
\softO\left(r^\Omega\cdot \left(\frac{d_\theta}{q-1} + d + c\right)^\Omega d \cdot \prec\right)
\]
operations in $\bbF$. The announced complexity follows after observing that
$c = O(\log \prec)$.
\end{proof}

The rather surprising fact that we can achieve quasilinear complexity with 
respect to the precision is a consequence of the extremely rapid $v$-adic 
convergence of the $L$-series $L_{v}(\underline{M},T)$: in order to get an 
accurate result at precision $v^{\prec}$, one only needs $O(\log \prec)$ 
terms in the expansion of the series.
This amazing property can be rephrased in analytic terms as follows.

\begin{Theorem}\label{thm:l-adic-infinite-radius}
Let $v$ be a finite or infinite place of $\bbF[t]$.
Write $L_{v}(\underline{M},T)=\sum_{i}{a_n T^n}$ for coefficients $a_n\in A_{v}$. Then,
\[
|a_n|_{v}=O\left(q^{-\deg v \cdot c^n}\right)
\]

where $c=q^{1/(\operatorname{rank}\underline{M}\cdot \deg v)}>1$ and 
$|\cdot|_{v}=q^{-\deg v\cdot v_v(\cdot)}$.
In particular, the radius of convergence of $L_{v}(\underline{M},T)$ is
infinite.
\end{Theorem}

Another remarkable fact is that the parameter $h$ does not appear
in the complexity. This is again the algorithmical counterpart of 
an interesting theoretical result of continuity. Before stating it,
we define the Carlitz twist of a $t$-motive: if $\underline M = 
(M, \tau_M)$ is a $t$-motive and $h$ is an integer, we set 
$\underline M(h) = (M, (t-\theta)^{-h} \tau_M)$.

\begin{Theorem}
Let $v$ be a finite place of $\bbF[t]$ of degree $d$.
Let $h, h'$ be two integers such that $h \equiv h' \pmod {(q^d-1) q^c}$
for some nonnegative integer $c$. Then
\[
L_v\big(\underline{M}(h), T\big) \equiv
L_v\big(\underline{M}(h'), T\big) \pmod{v^{q^c}}.
\]
\end{Theorem}

\begin{proof}
This follows from the observation that our algorithm computes the
$L$-series $L_v\big(\underline{M}(h), T\big)$ modulo $v^{q^c}$,
by extracting from $h$ only the following values:
\[
q h_1 - h_0,\,\, qh_2 - h_1,\,\,\ldots,\,\, q h_c - h_{c-1},\,\,
\left\lceil \frac{h_c}{q^d-1} \right\rceil,\,\,
\left\lceil \frac{qh_c}{q^d-1} \right\rceil,\,\, \ldots,\,\,
\left\lceil \frac{q^{d-1}h_c}{q^d-1} \right\rceil.
\]
One checks that all the preceding values depend only on the class 
of $h$ modulo $(q^d-1)q^c$.
\end{proof}

\subsubsection{A remark on the maximal model}

The algorithm of Figure~\ref{algo:Lseries} assumes that the maximal
model $M_{\sO}$ of $\underline M$ is known.
If it is not the case, one can compute it using the algorithm of 
Subsection \ref{ssec:maximalmodel}. However, this method could be quite 
slow as the description of the maximal model we obtain this way could 
involve polynomials in $t$ and $\theta$ of very large degrees.

A better option is the following. We start with a model $M_{\text{init}}$,
which might be not maximal.
We compute its discriminant $\Delta_{M_{\text{init}}} \in \bbF[\theta]$ 
and factor it: $\Delta_{M_{\text{init}}} = \fp_1^{n_1} \cdots \fp_s^{n_s}$.
For each place $\fp_i$, we now compute the \emph{local} maximal model
$M_{\sO, \fp_i}, \ldots, M_{\sO, \fp_s}$ using the method sketched 
in \S \ref{sssec:algo:localmaximalmodel} and the corresponding local
factor $P_{\fp_i}(M_{\sO, \fp_i},T)$.
The desired $L$-series can finally be recovered using the formula:
\[
L_v(\underline M, T) = 
L_v(M_{\text{init}}, T) \cdot \prod_{i=1}^s 
  \frac{P_{\fp_i}(M_{\text{init}}, T)}{P_{\fp_i}(M_{\sO, \fp_i},T)}
\]
The $L$-series $L_v(M_{\text{init}}, T)$ can be computed using
Algorithm~\ref{algo:Lseries} (with $M_{\sO}$ replaced by
$M_{\text{init}}$) while finding the local factors amounts to
computing the characteristic polynomial of explicit matrices.

\section{Implementation and timings}

We have implemented in SageMath~\cite{sage} the algorithm of
Figure~\ref{algo:Lseries} and the algorithm of
\S \ref{sssec:algo:localmaximalmodel} for computing local 
maximal models.
Combining them, we can compute the $L$-series of any $t$-motive
regardless of its reduction properties.
Our package is publicly available on gitlab at the address

\begin{center}
\url{https://plmlab.math.cnrs.fr/caruso/anderson-motives}
\end{center}

\begin{figure}
  \begin{tikzpicture}[scale=1.5]
\draw[-latex] (1.597,-3.175)--(8.597,-3.175);
\node[right, scale=0.7] at (8.597,-3.175) { precision };
\draw (1.297,-2.875)--(1.297,3.774);
\draw (1.257,-2.000)--(1.337,-2.000);
\node[left, scale=0.7] at (1.257,-2.000) { $10$ms };
\draw[thin, dotted, black!50] (1.847, -2.000)--(8.347, -2.000);
\draw (1.257,-1.000)--(1.337,-1.000);
\node[left, scale=0.7] at (1.257,-1.000) { $100$ms };
\draw[thin, dotted, black!50] (1.847, -1.000)--(8.347, -1.000);
\draw (1.257,0.000)--(1.337,0.000);
\node[left, scale=0.7] at (1.257,0.000) { $1$s };
\draw[thin, dotted, black!50] (1.847, 0.000)--(8.347, 0.000);
\draw (1.257,1.000)--(1.337,1.000);
\node[left, scale=0.7] at (1.257,1.000) { $10$s };
\draw[thin, dotted, black!50] (1.847, 1.000)--(8.347, 1.000);
\draw (1.257,2.000)--(1.337,2.000);
\node[left, scale=0.7] at (1.257,2.000) { $100$s };
\draw[thin, dotted, black!50] (1.847, 2.000)--(7.597, 2.000);
\draw (1.257,3.000)--(1.337,3.000);
\node[left, scale=0.7] at (1.257,3.000) { $1000$s };
\draw[thin, dotted, black!50] (1.847, 3.000)--(7.597, 3.000);
\begin{scope}[xshift=8.097cm, yshift=2.774cm]
\begin{scope}[color=black]
\fill (-0.040, -0.040) rectangle (0.040,0.040);
\fill (0.260, -0.040) rectangle (0.340,0.040);
\draw (0, 0.000)--(0.3, 0.000);
\node[right,scale=0.9] at (0.4, 0.000) { place at infinity };\end{scope}
\begin{scope}[color=blue]
\fill (-0.040, 0.260) rectangle (0.040,0.340);
\fill (0.260, 0.260) rectangle (0.340,0.340);
\draw (0, 0.300)--(0.3, 0.300);
\node[right,scale=0.9] at (0.4, 0.300) { place of degree $1$ };\end{scope}
\begin{scope}[color=red]
\fill (-0.040, 0.560) rectangle (0.040,0.640);
\fill (0.260, 0.560) rectangle (0.340,0.640);
\draw (0, 0.600)--(0.3, 0.600);
\node[right,scale=0.9] at (0.4, 0.600) { place of degree $2$ };\end{scope}
\end{scope}
\begin{scope}[color=black]
\fill (2.057, -2.415) rectangle (2.137,-2.335);
\fill (2.960, -2.353) rectangle (3.040,-2.273);
\fill (3.863, -2.308) rectangle (3.943,-2.228);
\fill (5.057, -2.035) rectangle (5.137,-1.955);
\fill (5.960, -1.553) rectangle (6.040,-1.473);
\fill (6.863, -1.037) rectangle (6.943,-0.957);
\fill (8.057, -0.573) rectangle (8.137,-0.493);
\draw (2.097, -2.375)--(3.000, -2.313)--(3.903, -2.268)--(5.097, -1.995)--(6.000, -1.513)--(6.903, -0.997)--(8.097, -0.533);
\node[right,scale=0.7] at (8.137, -0.533) { rank $2$ };\fill (2.057, -2.132) rectangle (2.137,-2.052);
\fill (2.960, -2.079) rectangle (3.040,-1.999);
\fill (3.863, -1.942) rectangle (3.943,-1.862);
\fill (5.057, -1.730) rectangle (5.137,-1.650);
\fill (5.960, -1.247) rectangle (6.040,-1.167);
\fill (6.863, -0.760) rectangle (6.943,-0.680);
\fill (8.057, -0.207) rectangle (8.137,-0.127);
\draw (2.097, -2.092)--(3.000, -2.039)--(3.903, -1.902)--(5.097, -1.690)--(6.000, -1.207)--(6.903, -0.720)--(8.097, -0.167);
\node[right,scale=0.7] at (8.137, -0.167) { rank $3$ };\fill (2.057, -1.921) rectangle (2.137,-1.841);
\fill (2.960, -1.897) rectangle (3.040,-1.817);
\fill (3.863, -1.775) rectangle (3.943,-1.695);
\fill (5.057, -1.430) rectangle (5.137,-1.350);
\fill (5.960, -0.931) rectangle (6.040,-0.851);
\fill (6.863, -0.457) rectangle (6.943,-0.377);
\fill (8.057, -0.072) rectangle (8.137,0.008);
\draw (2.097, -1.881)--(3.000, -1.857)--(3.903, -1.735)--(5.097, -1.390)--(6.000, -0.891)--(6.903, -0.417)--(8.097, -0.032);
\node[right,scale=0.7] at (8.137, -0.032) { rank $4$ };\end{scope}
\begin{scope}[color=blue]
\fill (2.057, -1.717) rectangle (2.137,-1.637);
\fill (2.960, -1.228) rectangle (3.040,-1.148);
\fill (3.863, -1.137) rectangle (3.943,-1.057);
\fill (5.057, -0.752) rectangle (5.137,-0.672);
\fill (5.960, -0.112) rectangle (6.040,-0.032);
\fill (6.863, 0.051) rectangle (6.943,0.131);
\fill (8.057, 0.273) rectangle (8.137,0.353);
\draw (2.097, -1.677)--(3.000, -1.188)--(3.903, -1.097)--(5.097, -0.712)--(6.000, -0.072)--(6.903, 0.091)--(8.097, 0.313);
\node[right,scale=0.7] at (8.137, 0.313) { rank $2$ };\fill (2.057, -1.231) rectangle (2.137,-1.151);
\fill (2.960, -0.585) rectangle (3.040,-0.505);
\fill (3.863, -0.538) rectangle (3.943,-0.458);
\fill (5.057, -0.067) rectangle (5.137,0.013);
\fill (5.960, 0.681) rectangle (6.040,0.761);
\fill (6.863, 0.758) rectangle (6.943,0.838);
\fill (8.057, 0.696) rectangle (8.137,0.776);
\draw (2.097, -1.191)--(3.000, -0.545)--(3.903, -0.498)--(5.097, -0.027)--(6.000, 0.721)--(6.903, 0.798)--(8.097, 0.736);
\node[right,scale=0.7] at (8.137, 0.736) { rank $3$ };\fill (2.057, -0.699) rectangle (2.137,-0.619);
\fill (2.960, -0.152) rectangle (3.040,-0.072);
\fill (3.863, 0.189) rectangle (3.943,0.269);
\fill (5.057, 0.522) rectangle (5.137,0.602);
\fill (5.960, 1.416) rectangle (6.040,1.496);
\fill (6.863, 1.476) rectangle (6.943,1.556);
\fill (8.057, 1.779) rectangle (8.137,1.859);
\draw (2.097, -0.659)--(3.000, -0.112)--(3.903, 0.229)--(5.097, 0.562)--(6.000, 1.456)--(6.903, 1.516)--(8.097, 1.819);
\node[right,scale=0.7] at (8.137, 1.819) { rank $4$ };\end{scope}
\begin{scope}[color=red]
\fill (2.057, -0.211) rectangle (2.137,-0.131);
\fill (2.960, 0.682) rectangle (3.040,0.762);
\fill (3.863, 1.251) rectangle (3.943,1.331);
\fill (5.057, 1.827) rectangle (5.137,1.907);
\fill (5.960, 2.199) rectangle (6.040,2.279);
\fill (6.863, 3.234) rectangle (6.943,3.314);
\draw (2.097, -0.171)--(3.000, 0.722)--(3.903, 1.291)--(5.097, 1.867)--(6.000, 2.239)--(6.903, 3.274);
\node[right,scale=0.7] at (6.943, 3.274) { rank $2$ };\fill (2.057, 0.736) rectangle (2.137,0.816);
\fill (2.960, 1.552) rectangle (3.040,1.632);
\fill (3.863, 2.035) rectangle (3.943,2.115);
\fill (5.057, 2.758) rectangle (5.137,2.838);
\draw (2.097, 0.776)--(3.000, 1.592)--(3.903, 2.075)--(5.097, 2.798);
\node[right,scale=0.7] at (5.137, 2.798) { rank $3$ };\fill (2.057, 0.981) rectangle (2.137,1.061);
\fill (2.960, 1.986) rectangle (3.040,2.066);
\fill (3.863, 2.593) rectangle (3.943,2.673);
\draw (2.097, 1.021)--(3.000, 2.026)--(3.903, 2.633);
\node[right,scale=0.7] at (3.943, 2.633) { rank $4$ };\end{scope}
\draw (2.097,-3.215)--(2.097,-3.135);
\node[below, scale=0.7] at (2.097,-3.215) { $5$ };
\draw (3.000,-3.215)--(3.000,-3.135);
\node[below, scale=0.7] at (3.000,-3.215) { $10$ };
\draw (3.903,-3.215)--(3.903,-3.135);
\node[below, scale=0.7] at (3.903,-3.215) { $20$ };
\draw (5.097,-3.215)--(5.097,-3.135);
\node[below, scale=0.7] at (5.097,-3.215) { $50$ };
\draw (6.000,-3.215)--(6.000,-3.135);
\node[below, scale=0.7] at (6.000,-3.215) { $100$ };
\draw (6.903,-3.215)--(6.903,-3.135);
\node[below, scale=0.7] at (6.903,-3.215) { $200$ };
\draw (8.097,-3.215)--(8.097,-3.135);
\node[below, scale=0.7] at (8.097,-3.215) { $500$ };
\end{tikzpicture}

  \caption{Timings for the computation of the $L$-series with $q = d_\theta = 9$.}
  \hfill%
  {\small (CPU: Intel Core i5-8250U at 1.60GHz --- OS: Ubuntu 22.04.1 --- SageMath 10.3)}%
  \hfill\null
  \label{fig:benchmarks}
\end{figure}

Figure~\ref{fig:benchmarks} gives an overview on the timings obtained with 
our package in the good reduction case when the precision, the rank and the degree of the place vary.
For fast computations (less than 10 seconds), the timings displayed in the
figure are averaged on many runnings (up to 100) of the algorithm with 
various inputs. They are therefore quite reliable. On the contrary, slow
computations were run only once; hence the reader should consider the
corresponding timings with extreme caution.

We observe that the dependence in the precision is rather good, almost
quasilinear as predicted by Theorem~\ref{thm:complexity}. 
It is also striking that the degree of the underlying place has a strong 
impact on the timings, much higher than the rank of the $t$-motive. 
Again, this is in line with the theoretical complexity given in 
Theorem~\ref{thm:complexity} although the effect seems even more pronounced 
on the timings. This could be due to the fact that computations in the 
completions $A_{v}$ stays slow.

We notice also that the computation for the place at infinity is much faster 
(by a factor $10$ about) than for any other place of degree~$1$. This can be 
explained by the fact that, when working at the place at infinity, the size 
of the nucleus can be lowered a bit according to the $t$-degrees of the entry 
of $\Phi$, an optimization we have implemented in our package.

Finally, we underline that, in all cases, most of the time is spent in the 
final computation of the characteristic polynomial, for which we rely on the 
implementation included in SageMath. Unfortunately, the latter is not 
optimized (SageMath uses a quartic algorithm) and is rather slow. 
Improving this basic routine would then automatically result in a
significant speed up of the computation of the $L$-series.

\medskip

\hrulefill\hfill\hfill\hfill\hfill\hfill\hfill

The authors have no competing interests to declare that are relevant to 
the content of this article.

The dataset supporting Conjecture~\ref{conjecture} is available at\\
\url{https://plmlab.math.cnrs.fr/caruso/anderson-motives/-/tree/main/examples}

\def\cprime{$'$}

\end{document}